\documentclass[11pt]{article}
\usepackage{amsmath}
\usepackage{graphicx,psfrag,epsf,multirow}
\usepackage{enumerate}
\usepackage{natbib}
\usepackage{url}

\usepackage{color}
\usepackage{booktabs,setspace}

\usepackage{amsmath, amsthm, amssymb}
\newtheorem{thm}{Theorem}
\newtheorem{lem}{Lemma}

\newtheorem{cor}{Corollary}
\theoremstyle{definition}

\renewcommand{\baselinestretch}{\vv}

\title{Dynamic Linear Discriminant Analysis in High Dimensional Space}
\author{Binyan Jiang, \and Ziqi Chen, \and Chenlei Leng
  \footnote{Jiang is with the Hong Kong Polytechnic University. Chen is with Central South University. Leng is with
    the University of Warwick. 
    }
}

\begin{document}
 \renewcommand{\baselinestretch}{1.2}

\maketitle

\begin{abstract}
High-dimensional data that evolve dynamically feature predominantly in the modern data era.  As a partial response to this,
recent years have seen increasing emphasis to
address the dimensionality challenge. However, the non-static nature of these datasets is largely
ignored. This paper addresses both challenges by proposing a novel yet simple dynamic linear
programming discriminant (DLPD) rule for binary classification. Different from the usual static linear discriminant analysis, the new method is able to capture the changing distributions of the underlying populations by modeling their means and covariances as smooth functions of covariates of interest. Under an approximate sparse condition, we show that the conditional misclassification rate of the DLPD rule converges to the Bayes risk in probability {\it uniformly} over the range of the variables used for modeling the dynamics, when the dimensionality is allowed to grow exponentially with the sample size. The minimax lower bound of the estimation of the Bayes risk is also established, implying that the misclassification rate of our proposed rule is minimax-rate optimal. The promising performance of the DLPD rule is illustrated via extensive simulation studies and the analysis of a breast cancer dataset.
\end{abstract}

\noindent
{\bf Key Words:} 
{\it Bayes rule; Discriminant analysis; Dynamic linear programming; High-dimensional data; Kernel smoothing; Sparsity.}

 \section{Introduction}
 The rapid development of modern measurement technologies has enabled us to gather data that are increasingly larger. As the rule rather than the exception, these datasets have been gathered at different time,  under different conditions, subject to a variety of perturbations, and so on. As a result,
 the complexity of many modern data is predominantly characterized by high dimensionality and the data dynamics. The former is featured by a large number of variables in comparison to the sample size, and the manifestation of the latter can be seen in the distribution of the data which is non-static and dependent on covariates such as time. Any approach ignoring either of the two aspects may give unsatisfactory performance and even incorrect conclusions.
 
 The main aim of this paper is to address these two challenges simultaneously, for the first time,
 by developing a very simple yet useful dynamic linear programming discriminant (DLPD) rule for classification. Specializing to
 binary classification, we allow the means and the covariance matrices of the
 populations to vary with covariates of interest, which are estimated via local
 smoothing \citep{Fan1996}. Under an approximate sparsity assumption on a linear index that is central to classification, we propose to
 estimate the index vector via a technique akin to the Dantzig
 selector \citep{candes,Cai} in a dynamic setting. We show emphatically that
 the conditional misclassification rate of the DLPD rule converges to the Bayes
 risk in probability {\it uniformly} over a range of the variables used for
 modeling dynamics, where the dimensionality is allowed to be exponentially high relative to the sample size.
 The uniformity result is of particular importance as it
 permits simultaneous statements over the whole range of the covariate. In addition, we derived minimax lower bounds for the Bayes risk, which indicates that the misclassification rate of our DLPD rule is minimax-rate optimal. To our
 best knowledge, this is the first attempt in developing a high-dimensional
 discriminant method that exhibits local features of the data with sound theory. We remark that using existing approaches such as the one in \cite{Cai} coupled with local smoothing, it is possible to establish a {\it pointwise} result for the misclassification rate. However, a pointwise convergence result will not be sufficient in a dynamic setting, as the main interest is often to assess the estimated classification rule across the whole of the covariates, not just at a single point of the covariates.

 Before we proceed further, let's quickly look at a dataset that motivated this study. In traditional disease diagnosis studies, the same classification rule for all the patients was often applied. However, it has become increasingly more desirable to develop personalized rules that takes into account individual characteristics \citep{pm2015}. Intuitively, these patient-specific factors can be treated as dynamic factors in deriving decision rules. For example,
 in the breast cancer data we studied in Section 4.3, both (low dimensional) clinical risk factors (tumor size, age, histological grade etc.) and (high dimensional) expression levels for 24,481 gene probes were collected for 97 lymph node-negative breast cancer patients. Among them, 46 patients developed distant metastases within 5 years while the rest 51  remained metastases free for at least 5 years.
 To appreciate the need to incorporate dynamic information into the analysis, we look at the 100 genes with the largest absolute $t$-statistic values between the two groups choosing the tumor size as the dynamic variable. We fit the gene expression levels as a function of the tumor size using
 a local regression model \citep{Cleveland}. The fitted plots for some randomly selected genes are presented in Figure \ref{figgene}, from which we can see that the gene expression levels of the patients in the two classes exhibit different levels as the tumor size changes. Similarly, the covariance matrix of these 100 genes also is found to behave dynamically
 in response to the changes of the tumor size. To see this, we separate the
 97 observations into two groups depending on whether the tumor size
 is greater than the median of the tumor sizes 2.485. A $p$-value $ <0.001$ \citep{Li2012} indicates that we should reject
 the null hypothesis that the population covariance matrices of the two groups
 are equal.
 The method developed in this paper aims to capture this dynamic
 information in a high-dimensional setting for classification.
 
 \begin{figure}[htbp]
 	\centering
 	\includegraphics[width=0.8\textwidth]{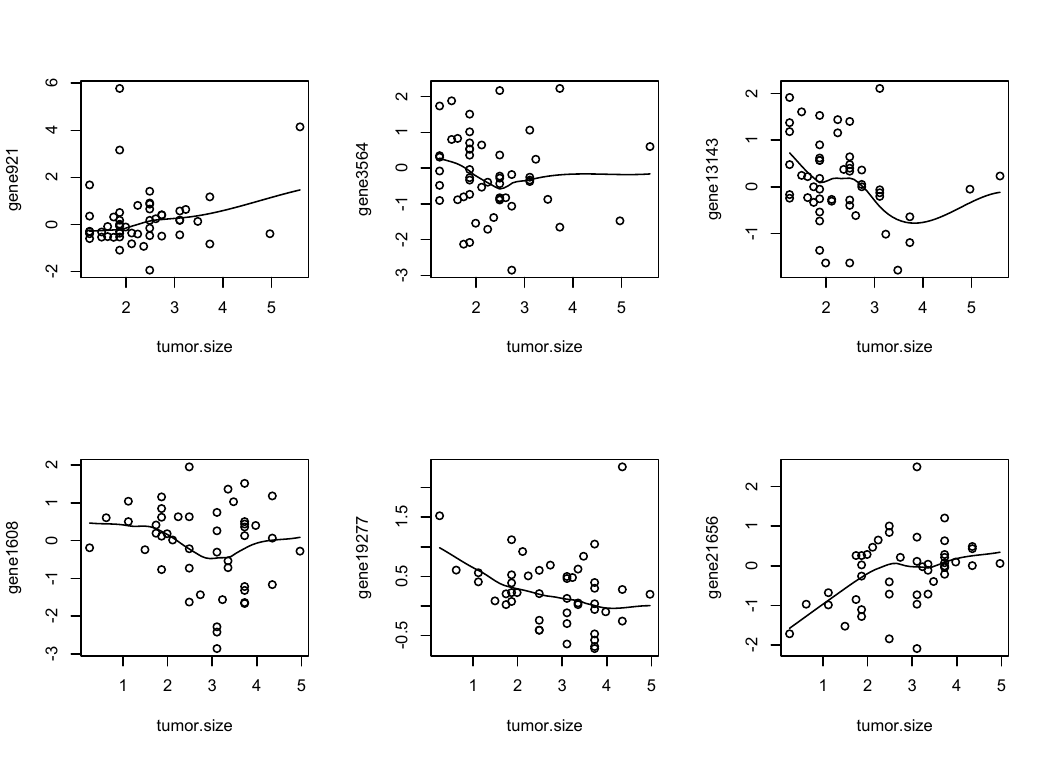}
 	\caption{Gene expression level versus tumor size. Upper panel: selected genes from $X$ class; Lower panel: selected genes from $Y$ class. The curves are LOWESS fits.}\label{figgene}
 \end{figure}

 \subsection{The setup}
 We now introduce formally the problem.
 Let $X=(x_1,\ldots,$ $x_p)^T$, $Y=(y_1,\ldots,y_p)^T$ be $p$-dimensional
 random vectors and ${\bf U}$ be a $d$-dimensional random covariate, where for simplicity we assume that $d$ is a fixed integer. In this paper we deal with the situation where $p$ is large.
 Given ${\bf U}$ we
 assume that $X\sim N(\mu_X({\bf U}), \Sigma({\bf U}))$ where $\mu_X({\bf U}) =
 (\mu_{X}^1({\bf U}),\ldots,\mu_{X}^p({\bf U}))^T$ and $\Sigma({\bf {\bf U}})=(\sigma_{ij}({\bf U}))_{1\leq
 	i,j \leq p}$. Similarly, the conditional distribution of $Y$ given ${\bf U}$ is
 given as $Y\sim N(\mu_Y({\bf U}), \Sigma({\bf U}))$ where $\mu_Y({\bf U}) =
 (\mu_{Y}^1({\bf U}),\ldots,\mu_{Y}^p({\bf U}))^T$. In other words, different from
 traditional linear discriminant analysis, we assume that the first and second moments of $X$ and $Y$ change over a $d$-dimensional covariate ${\bf U}$.
 Here ${\bf U}$ could be dependent on the features $X$ and $Y$.
 When ${\bf U}$ is a vector of discrete variables, the above mentioned model is named the location-scale model and was used for discriminant analysis with mixed data under finite dimension assumptions; see, for example, \cite{Krz1993} and the references therein.
 
 In discriminant analysis, it is well known that the Bayes procedure is
 admissible; see for example \cite{Anderson}. Let $(Z, {\bf U}_Z)$ be a generic random sample which can be
 from either the population $(X; {\bf U})$ or the population $(Y; {\bf U})$. In this paper
 we assume a priori that it is equally likely that  $(Z, {\bf U}_Z)$ comes from
 either population $(X; {\bf U})$  or population $(Y; {\bf U})$. Following simple algebra, it can be easily shown
 that the Bayes procedure is given as the following:
 
 \begin{itemize}
 	\item[(i)] Classify $(Z, {\bf U}_Z)$ into population $(X, {\bf U})$ if
 	\[\{Z-[\mu_X({\bf U}_Z)+\mu_Y({\bf U}_Z)]/2\}^T\Sigma^{-1}({\bf U}_Z)[\mu_X({\bf U}_Z)-\mu_Y({\bf U}_Z)]\geq 0;\]
 	\item[(ii)] Classify $(Z, {\bf U}_Z)$ into population $(Y, {\bf U})$ otherwise.
 \end{itemize}
 
 Given ${\bf U}_z={\bf u}$, by standard calculation, the conditional misclassification rate of this rule is
 \begin{eqnarray}\label{R}
 	R({\bf u})=\Phi(-\Delta_p({\bf u})/2),
 \end{eqnarray}
 where $\Delta_p({\bf u})=\sqrt{[\mu_X({\bf u})-\mu_Y({\bf u})]^T\Sigma^{-1}({\bf u})[\mu_X({\bf u})-\mu_Y({\bf u})]},$
 and $\Phi(\cdot)$ is the cumulative distribution function of a standard normal random variable. The expected misclassification rate is defined as
 \begin{eqnarray}\label{ER}
 	R={\rm E}_{\bf U}R({\bf U}),
 \end{eqnarray}
 where ${\rm E}_{\bf U}$ means taking expectation with respect to ${\bf U}$.
 Practically $\mu_X(\cdot)$, $\mu_Y(\cdot)$ and $\Sigma(\cdot)$ are unknown but there are a sequence of independent random observations $(X_i,{\bf U}_i)$, $i=1,\ldots, n_1$ from the population $(X; {\bf U})$ and a sequence of independent random observations $(Y_j, {\bf V}_j)$, $j=1\ldots, n_2$ from the population $(Y, {\bf U})$.  The central problem then becomes proposing methods based on the sample that give misclassification rates converging to that of the Bayes rule under appropriate assumptions.
 
 \subsection{Existing works}
 There has been increasing emphasis in recent years to address the
 high-dimensionality challenge posed by modern data where $p$ is large. However, the dynamic
 nature of the data collection process is often ignored in that $\mu({\bf U})$ and
 $\Sigma({\bf U})$ are assumed to be independent of ${\bf U}$.
 In this static case, the Bayes procedure given above reduces to the well-known Fisher's linear discriminant analysis (LDA).
 In static high dimensional discriminant analysis, \cite{Bickel2004} first highlighted that Fisher's LDA is equivalent to random guessing. Fortunately, in many
 problems, various quantities in the LDA can be assumed sparse; See, for example, \cite{Witten:Tibs:2011,Shao,Cai,Fan2012,Maiqing2012}, and \cite{Mai:Zou:2013} for a summary of selected sparse LDA methods. Further studies along this line can be found in \cite{Fan:etal:2013} and \cite{Hao:etal:2015}. More recently,  quadratic discriminant analysis has attracted increasing attention where the population covariance matrices are assumed static but different. This has motivated the study of more flexible models exploiting variable interactions for classification, analogous to two-way interaction in linear regression; see for example \cite{Fan:etal:2015}, \cite{Fan:etal:2015a}, and \cite{Jiang:etal:2015}. However, none of these works addresses the dynamic nature of $\mu(\cdot)$ and $\Sigma(\cdot)$.

 In our setup where dynamics exists, in addition to the high dimensionality, we need to obtain dynamic estimators for
 $\mu_X({\bf u})-\mu_Y({\bf u})$ and $\Sigma^{-1}({\bf u})$, or
 \[\beta({\bf u}):=\Sigma^{-1}({\bf u})[\mu_X({\bf u})-\mu_Y({\bf u})]\] as functions of ${\bf u}$. Under a similar setup where ${\bf U}$ is categorical and supported on a set of finite elements,  \cite{Guo} proposed a sparse estimator for
 $\Sigma^{-1}({\bf u})$. The emphasis of this work is for continuous ${\bf U}$ that is compactly supported. \cite{Ziqi} proposed
 nonparametric estimators of sparse $\Sigma({\bf u})$ using thresholding
 techniques for univariate ${\bf U}$ where $d=1$. The focus of this paper on
 high-dimensional classification is completely different. Importantly, we do
 not require the sparsity assumption on  $\Sigma({\bf u})$ and our theory applies for any fixed-dimensional ${\bf U}$. Our paper is also different from \cite{Cai},
 \cite{Fan2012} and \cite{Maiqing2012} in that $\beta(\bf u)$ is allowed not only to be a smooth function of ${\bf U}$, but also to be approximately sparse (see Theorem \ref{theorem1}). Our efforts greatly advance the now-classical approach of local polynomial \citep{Fan1996} to the modern era of high-dimensional data analysis.

 If we denote  $\hat{\mu}_X({\bf U}_Z), \hat{\mu}_Y({\bf U}_Z)$ and $\hat{\beta}({\bf U}_Z)$ as
 the estimators of $\mu_X({\bf U}_Z)$, $\mu_Y({\bf U}_Z)$ and $\beta({\bf U}_Z)$ defined as in
 Section 2 respectively, our Dynamic Linear Programming Discriminant (DLPD)
 rule is given as the following:
 
 \begin{itemize}
 	\item[(i)] Classify $(Z, {\bf U}_Z)$ into population $(X, {\bf U})$ if:
 	\[\{Z-[\hat{\mu}_X({\bf U}_Z)+\hat{\mu}_Y({\bf U}_Z)]/2\}^T\hat{\beta}({\bf U}_Z)\geq 0;\]
 	\item[(ii)] Classify $(Z, {\bf U}_Z)$ into population $(Y, {\bf U})$ otherwise.
 \end{itemize}
 
 The rest of this paper is organized as follows. In Section 2, we propose
 estimators for the components in the Bayes rule and propose the DLPD rule. Section 3 provides theoretical results of our DLPD rule. In particular, we show that under appropriate conditions,
 the risk function of the DLPD rule converges to the Bayes risk function
 uniformly in ${\bf u}$. In addition, we derived minimax lower bounds for the estimation of $\Delta({\bf u})$ and the Bayes risk. In section 4, simulation study is conducted to assess the finite sample performance of the proposal method. The DLPD rule is then applied to solve interesting
 discriminant problems using 
 a breast cancer dataset.  Concluding remarks are made in Section 5. All the theoretical proofs are given in the Appendix.

 \section{A dynamic linear programming discriminant rule}
 We begin by introducing some notations. For any matrix $M$, we use
 $M^T$, $|M|$ and ${\rm tr}(M)$ to denote its transpose, determinant and trace. 
 Let $v=(v_1,\ldots,v_p)^T\in
 \mathcal{R}^p$ be a $p$-dimensional vector. Define
 $|v|_0=\sum_{i=1}^pI_{\{v_i\neq 0\}}$ as the $\ell_0$ norm and $|v|_\infty=\max_{1\leq
 	i\leq p}|v_i|$ as the $\ell_\infty$ norm. For any $1\leq q<\infty$, the $l_q$ norm of $v$ is
 defined as $|v|_q=(\sum_{i=1}^p|v_i|^q)^{1/q}$. We denote the $p$-dimensional
 vector of ones as $1_p$ and the $p$-dimensional vector
 of zeros as $0_p$.
 
 Denote ${\bf u}=(u_1,\ldots, u_d)^T$ and let ${\bf K}({\bf u})$ be a kernel function such that
 \begin{eqnarray*}
 	{\bf K}({\bf u})=\Pi_{i=1}^dK(u_1)\times\cdots\times K(u_d),
 \end{eqnarray*}
 where $K(\cdot)$ is an univariate kernel function,  for example, the Epanechnikov kernel used in kernel
 smoothing \citep{Fan1996}. Recent literature on multivariate kernel estimation can be found in \cite{Gu2015} and the references therein. Let $H={\rm diag}\{h_1,\ldots, h_d\}$ be a $d\times d$ diagonal bandwidth matrix and define:
 \begin{eqnarray*}
 	{\bf K}_H({\bf u})=|H|^{-1}{\bf K}(H^{-1}{\bf u})=\Pi_{i=1}^{d}\frac{1}{h_i} K\left(\frac{u_i}{h_i}\right).
 \end{eqnarray*}
 
 Recall that we assume that there are a sequence of independent random observations $(X_i,{\bf U}_i)$, $i=1,\ldots, n_1$, from the population $(X; {\bf U})$ and a sequence of independent random observations $(Y_j, {\bf V}_j)$, $j=1\ldots, n_2$, from the population $(Y, {\bf U})$. For simplicity, throughout this paper we assume that  $n_1\asymp n_2$ and denote $n=n_1+n_2$.
 
 One of the most popular nonparametric estimators for estimating a conditional expectation is the Nadaraya-Watson estimator, which is a locally weighted average, using a kernel as a weighting function.  Denote $X_i=(X_{i1},\ldots,X_{ip})^T, i=1,\ldots,n_1$. Let $H_x={\rm diag}\{h_{x1},\ldots, h_{xd}\}$ be a given bandwidth matrix. We estimate $\mu_X({\bf u})$ using the Nadaraya-Watson estimator \citep{Nadaraya} $\hat{\mu}_X({\bf u})=(\hat{\mu}^1_X({\bf u}),\ldots,\hat{\mu}^p_X({\bf u}))^T$, where
 \begin{eqnarray}\label{muhatx}
 	\hat{\mu}^{i}_X({\bf u})=\frac{\sum_{j=1}^{n_1}{\bf K}_{H_x}({\bf U}_j-{\bf u})X_{ji}}{\sum_{j=1}^{n_1}{\bf K}_{H_x}({\bf U}_j-{\bf u})}, ~~~i=1,\ldots,p.
 \end{eqnarray}
 Similarly,  let $Y_i=(Y_{i1},\ldots,Y_{ip})^T, i=1,\ldots, n_2$. Given a  bandwidth matrix $H_y={\rm diag}\{h_{y1},\ldots, h_{yd}\}$,
 we estimate $\mu_Y({\bf u})$ by $\hat{\mu}_Y({\bf u})=(\hat{\mu}^1_Y({\bf u}),\ldots,\hat{\mu}^p_Y({\bf u}))^T$, where
 \begin{eqnarray}\label{muhaty}
 	\hat{\mu}^{i}_Y({\bf u})=\frac{\sum_{j=1}^{n_2}{\bf K}_{H_y}({\bf V}_j-{\bf u})Y_{ji}}{\sum_{j=1}^{n_2}{\bf K}_{H_y}({\bf V}_j-{\bf u})}, ~~~i=1,\ldots,p.
 \end{eqnarray}
 
 
 For the covariance matrix $\Sigma({\bf u})$, we propose the following empirical estimator:
 \begin{eqnarray}\label{sigmahat}
 	\hat{\Sigma}({\bf u})=(\hat{\sigma}_{ij}({\bf u}))_{1\leq i,j\leq p}=\frac{n_1}{n}\hat{\Sigma}_X({\bf u})+\frac{n_2}{n}\hat{\Sigma}_Y({\bf u}),
 \end{eqnarray}
 where
 \begin{eqnarray}\label{sigmahatx}
 	\hat{\Sigma}_X({\bf u})&=& \frac{\sum_{j=1}^{n_1}{\bf K}_{H_x}({\bf U}_j-{\bf u})X_{j}X_j^T}{\sum_{j=1}^{n_1}{\bf K}_{H_x}({\bf U}_j-{\bf u})}\\
 	&&-\frac{[\sum_{j=1}^{n_1}{\bf K}_{H_x}({\bf U}_j-{\bf u})X_{j}][\sum_{j=1}^{n_1}{\bf K}_{H_x}({\bf U}_j-{\bf u})X_{j}^T]}{[\sum_{j=1}^{n_1}{\bf K}_{H_x}({\bf U}_j-{\bf u})]^2}, \nonumber
 \end{eqnarray}
 and
 \begin{eqnarray}\label{sigmahaty}
 	\hat{\Sigma}_Y({\bf u})&=&\frac{\sum_{j=1}^{n_2}{\bf K}_{H_y}({\bf V}_j-{\bf u})Y_{j}Y_j^T}{\sum_{j=1}^{n_2}{\bf K}_{H_y}({\bf V}_j-{\bf u})}\\
 	&&-\frac{[\sum_{j=1}^{n_2}{\bf K}_{H_y}({\bf V}_j-{\bf u})Y_{j}][\sum_{j=1}^{n_2}{\bf K}_{H_y}({\bf V}_j-{\bf u})Y_{j}^T]}{[\sum_{j=1}^{n_2}{\bf K}_{H_y}({\bf V}_j-{\bf u})]^2}. \nonumber
 \end{eqnarray}
 We remark that the estimators $\hat{\mu}_X({\bf u}), \hat{\mu}_Y({\bf u}),
 \hat{\Sigma}_X({\bf u})$ and $\hat{\Sigma}_Y({\bf u})$ are simply the weighted sample estimates
 with weights determined by the kernel.
 
 For a given ${\bf u}$, we then estimate $\beta({\bf u})=\Sigma^{-1}({\bf u})[\mu_X({\bf u})-\mu_Y({\bf u})]$ using a Dantzig
 selector \citep{candes,Cai} as
 \begin{equation}\label{hatbeta}
 	\hat{\beta}({\bf u})={\rm arg min}_{\beta}\{|\beta|_1 {\rm~subject~to~}|\hat{\Sigma}({\bf u})\beta-[\hat{\mu}_X({\bf u})-\hat{\mu}_Y({\bf u})]|_{\infty}\leq \lambda_n\}.
 \end{equation}
 Given a new observation $(Z,{\bf U}_Z)$, our dynamic linear programming discriminant rule is obtained by plugging in the estimators given in (\ref{muhatx}), (\ref{muhaty}), (\ref{sigmahat}) and (\ref{hatbeta}) into
 the Bayes rule given in Section 1. That is,
 \begin{itemize}
 	\item[(i)] Classify $(Z, {\bf U}_Z)$ into population $(X, {\bf U})$ if:
 	\[\{Z-[\hat{\mu}_X({\bf U}_Z)+\hat{\mu}_Y({\bf U}_Z)]/2\}^T\hat{\beta}({\bf U}_z)\geq 0;\]
 	\item[(ii)] Classify $(Z, {\bf U}_Z)$ into population $(Y, {\bf U})$ otherwise.
 \end{itemize}

 \section{Theory}
 In this section we will first derive the theoretical properties of our proposed dynamic linear programming discriminant rule. In particular, the upper bounds of the misclassification rate are established.
 We will then derive minimax lower bounds for estimation of the misclassification rate. The upper bounds and lower bounds together show that the misclassification rate of our proposed discriminant rule achieves the optimal rate of convergence.
 
 \subsection{Upper bound analysis}
 
 In high dimensional data analysis, Bernstein-type inequalities are widely used
 to prove important theoretical results; see for example Lemma 4 of
 \cite{Bickel2004}, \cite{Merl}, Lemma 1 of \cite{Cai}. Different from
 existing literature in high dimensional linear discrimination analysis, we need to accommodate the dynamic pattern. Particularly, to prove our main results in this section, we establish uniform Bernstein-type inequalities for the mean estimators $\hat{\mu}_X({\bf u})$,  $\hat{\mu}_Y({\bf u})$ and the covariance matrix estimators $\hat{\Sigma}_X({\bf u})$ and $\hat{\Sigma}_Y({\bf u})$; see Lemma \ref{mean} and Lemma \ref{covariance}. We point out that these uniform concentration inequalities could be essential in other research problems that encounter high dimensionality and non-stationarity simultaneously. We present the risk function of the DLPD rule first.
 \begin{lem}\label{risk}
 	Let $\Omega_d\in {\mathcal R}^d$ be the support of ${\bf U}$ and ${\bf V}$.
 	Given ${\bf u}\in \Omega_d$, the conditional misclassification rate of the DLPD rule is
 	\begin{eqnarray*}
 		\hat{R}({\bf u})&=&\frac{1}{2}\Phi\Bigg(-\frac{(\hat{\mu}_X({\bf u})-\hat{\mu}_Y({\bf u}))^T\hat{\beta}({\bf u})}{2\sqrt{\hat{\beta}({\bf u})^T\Sigma({\bf u})\hat{\beta}({\bf u})}}-\frac{(\hat{\mu}_Y({\bf u})-\mu_Y({\bf u}))^T\hat{\beta}({\bf u})}{\sqrt{\hat{\beta}({\bf u})^T\Sigma({\bf u})\hat{\beta}({\bf u})}}\Bigg)\\&&
 		+
 		\frac{1}{2}\Phi\Bigg(-\frac{(\hat{\mu}_X({\bf u})-\hat{\mu}_Y({\bf u}))^T\hat{\beta}({\bf u})}{2\sqrt{\hat{\beta}({\bf u})^T\Sigma({\bf u})\hat{\beta}({\bf u})}}+\frac{(\hat{\mu}_X({\bf u})-\mu_X({\bf u}))^T\hat{\beta}({\bf u})}{\sqrt{\hat{\beta}({\bf u})^T\Sigma({\bf u})\hat{\beta}({\bf u})}}\Bigg).
 	\end{eqnarray*}
 \end{lem}
 
 To obtain our main theoretical results, we make the following assumptions.
 \begin{itemize}
 	\item[(A1)] The kernel function is symmetric in that $K(u)=K(-u)$ and there exists a constant $s>0$ such that $\int_{\mathcal R}
 	K(u)^{2+s}u^j{\rm d}u<\infty$ for $j=0,1,2$. In addition, there exists
 	constants $K_1$ and $K_2$ such that ${\rm sup}_{u\in\mathcal R}|K(u)|<K_1<\infty$
 	and ${\rm sup}_{u\in\mathcal R}|K'(u)|<K_2<\infty$.
 	
 	\item[(A2)] We assume the sample sizes $n_1\asymp n_2$ and denote $n=n_1+n_2$. In addition we assume that ${\frac{\log p}{n}}\rightarrow 0$ as $p, n \rightarrow \infty$ and for simplicity we also assume that $p$ is large enough such that $O(\log n +\log p) = O(\log p)$.

 	\item[(A3)]   $ {\bf U}_1,\ldots , {\bf U}_{n_1}, {\bf V}_1,\ldots, {\bf V}_{n_2}$ are independently and identically sampled from a distribution with a density function $f(\cdot)$, which has a compact support $\Omega_d\in {\mathcal R}^d$. In addition, $f(\cdot)$ is twice continuously differentiable and is bounded away from ${\bf 0}_d$ on its support.
 	
 	\item[(A4)]  The bandwidths satisfy $h_{xi}\asymp \left(\frac{\log p}{n_1}\right)^{\frac{1}{4+d}}$, $h_{yi}\asymp \left(\frac{\log p}{n_2}\right)^{\frac{1}{4+d}}$, for $i=1,\ldots, d.$
 	
 	\item[(A5)] Let $\lambda_1(\Sigma({\bf u}))$ and $\lambda_p(\Sigma({\bf u}))$ be the smallest and largest eigenvalues of $\Sigma({\bf u})$ respectively. We assume that 
 	There exists a positive constant $\lambda$ such that 
 	$\lambda^{-1} \leq \inf_{{\bf u}\in \Omega_d}\lambda_1(\Sigma({\bf u})) \leq \sup_{{\bf u}\in \Omega_d}\lambda_p(\Sigma({\bf u}))\leq \lambda$. In addition, there exists a constant $B>0$ such that $\inf_{{\bf u}\in \Omega_d}\Delta_p({\bf u}) > B$.
 	
 	%

 	\item[(A6)] The mean functions $\mu_X({\bf u}), \mu_Y({\bf u})$ and all the entries of $\Sigma({\bf u})$ have continuous second order derivatives in a neighborhood of each ${\bf u}$ belonging to the interior of $\Omega_d$.
 	
 	
 	Clearly, all the supremum and infimum in this paper can be relaxed to essential supremum and essential infimum.
 	
 \end{itemize}

 Assumptions (A1), (A3) and (A4) are commonly made on kernel functions in nonparametric
 smoothing literature; see for example \cite{Einmahl}, \cite{Fan1996} and \cite{Pagan}. The first statement of assumption (A2) is for simplicity and the
 second statement indicates that our approach allows the dimension $p$ to be as
 large as $O(\exp(n^{c}))$ for any constant $c<1$. That is, the
 dimensionality is allowed to be exponentially high in terms of the sample
 size.
 For assumption (A3), since the density function $f(\cdot)$ is continuous, the image set  ${\cal D}:=\{f({\bf u}): {\bf u}\in \Omega_d\}$ is also compact. Consequently, if there is a sequence of points $f_1,\ldots, f_m,\ldots$ that converges to 0, we must have $0\in {\cal D}$. Therefore our assumption that $f({\bf u})$ is bounded away from zero is equivalent to $f({\bf u})>0$ in ${\cal D}$. Note that the dominator $\sum_{j=1}^{n_1}{\bf K}_{H_x}({\bf U}_j-{\bf u})$ in the Nadaraya-Watson estimator converges to $f({\bf u})$. Our  assumption in some sense ensures that the dominator does not vanish. 
 We can though, relax the compactness condition on the support to the following: there exist $m$ compact sets $\Omega_{d,1},\ldots,\Omega_{d,m}\in {\mathcal R}^d$ such that for some constant $C_u>0$ and $M>0$ which is defined as in Theorems 3.1 and 3.2, we have $P({\bf U}\in \Omega_{d})\geq 1-C_up^{-M}$, where $ \Omega_{d}:=\cup_{i=1}^m\Omega_{d,i}$. 
 Assumption (A5) is routinely made in high dimensional discrimination analysis; see for example \cite{Cai}. Nevertheless, we may allow the uniform bounds on the eigenvalues to hold on $ \Omega_{d}:=\cup_{i=1}^m\Omega_{d,i}$, while assuming that $P({\bf U}\notin \Omega_{d})$ is negligible. 
 Assumption (A6) is a smoothness condition to ensure estimability and is commonly used in the literature of nonparametric estimation; see for example \cite{Fan1996,Tsybakov}.

 The following theorem shows that the risk function of the DLPD rule given
 in Lemma \ref{risk} converges to the Bayes risk function (\ref{R}) uniformly
 in ${\bf u}\in\Omega_d$.

 \begin{thm}\label{theorem1}
 	Assume that assumptions (A1)-(A6)
 	and the following assumption hold:
 	\begin{eqnarray}\label{assump0}
 		\sup_{{\bf u}\in \Omega_d}\frac{ |\beta({\bf u})|_1}{\Delta_p({\bf u})}=o\left(\left(\frac{n}{\log p}\right)^{\frac{2}{4+d}}\right).
 	\end{eqnarray}
 	For any constant $M>0$, by choosing 
 	$\lambda_n=C\left(\frac{\log p}{n}\right)^{\frac{2}{4+d}}\sup_{{\bf u}\in \Omega_d}\Delta({\bf u})$ for some constant $C$
 	large enough, we have with probability larger than $1-{O}(p^{-M})$,
 	\begin{eqnarray*}
 		\sup_{{\bf u}\in \Omega_d}|\hat{R}({\bf u})-R({\bf u})|=O\left( \left(\frac{\log p}{n}\right)^{\frac{2}{4+d}}\sup_{{\bf u}\in \Omega_d}\frac{|\beta({\bf u})|_1}{\Delta_p({\bf u})}\right).
 	\end{eqnarray*}
 	Consequently, we have
 	\begin{eqnarray*}
 		E_{\bf U}\hat{R}({\bf U})-R\rightarrow 0  ~~~{\rm as} ~~p, n\rightarrow \infty.
 	\end{eqnarray*}
 \end{thm}
 Here $\Delta_p({\bf u})$ measures the Mahalanobis distance between the
 two population centroids for a given ${\bf u}$. This theorem does not require $\beta({\bf u})$ to be
 sparse, but assumes the $\ell_1$ norm of $\beta({\bf u})$ divided by the
 Mahalanobis distance
 is bounded uniformly by a
 factor with an order smaller than $ \left(\frac{n}{\log p}\right)^{\frac{2}{4+d}}$. In
 particular, the dimensionality is allowed to diverge as quickly as
 $o(\exp(n))$. This theorem shows that {\it uniformly} in ${\bf U}$,
 the conditional misclassification rate converges to the Bayes risk in
 probability. In order
 to connect this theorem to the situation where $\beta({\bf u})$ is sparse, we note
 that from the Cauchy-Schwartz inequality and assumption (A5), we have for any ${\bf u}\in \Omega_d$,
 \begin{eqnarray*}
 	\frac{|\beta({\bf u})|_1^2}{\Delta_p^2({\bf u})}\leq
 	\frac{|\beta({\bf u})|_0|\beta({\bf u})|_2^2}{\Delta_p^2({\bf u})}\leq
 	\frac{|\beta({\bf u})|_0|\lambda^2|\mu_X({\bf u})-\mu_Y({\bf u})|_2^2}{\lambda^{-2}|\mu_X({\bf u})-\mu_Y({\bf u})|_2^2}=\lambda^4|\beta({\bf u})|_0.
 \end{eqnarray*}
 Consequently we have:
 \begin{cor}\label{cor} Assume that assumptions (A1)-(A6)
 	and the following assumption hold:
 	\begin{eqnarray}\label{assump}
 		\sup_{{\bf u}\in \Omega_d} |\beta({\bf u})|_0=o\left(\left(\frac{n}{\log p}\right)^{\frac{4}{4+d}} \right).
 	\end{eqnarray}
 	For any constant $M>0$, by choosing $\lambda_n=C\left(\frac{\log p}{n}\right)^{\frac{2}{4+d}}\sup_{{\bf u}\in \Omega_d}{\Delta_p({\bf u})}$ for some constant $C$
 	large enough, we have with probability larger than $1-{O}(p^{-M})$,
 	\begin{eqnarray*}
 		\sup_{{\bf u}\in \Omega_d}|\hat{R}({\bf u})-R({\bf u})|\rightarrow 0~~ {\rm and}~~ E_{\bf U}\hat{R}({\bf U})-R\rightarrow 0  ~~~{\rm as} ~~p, n\rightarrow \infty.
 	\end{eqnarray*}
 \end{cor}
 This corollary states that the conditional misclassification rate converges to the Bayes
 risk again, if the cardinality of $\beta({\bf u})$ diverges in an order
 smaller than $\left(\frac{n}{\log p}\right)^{\frac{4}{4+d}}$. Thus, our results apply to
 approximate sparse models as in Theorem \ref{theorem1} and sparse models as in
 Corollary \ref{cor}.
 
 In many high dimensional problems without  a dynamic variable ${\bf U}$,
 it has been commonly assumed that the dimension $p$ and sample size $n$
 satisfy $\frac{\log p}{n}\rightarrow 0$. Denote $H=H_x$ or
 $H_y$. From our proofs we see that in the dynamic case where ${\bf U}$ has an
 effect, due to the local estimation, the dimension-sample-size condition
 becomes $\frac{\log p{\rm tr}(H^{-1})}{n|H|}\rightarrow 0$, which becomes $\left(\frac{\log p}{n}\right)^{\frac{4}{4+d}}\rightarrow 0$ under Assumption (A4). We give here a heuristic
 explanation for the change in the dimension-sample-size condition when $d=1$. It is known
 that the variance of a kernel estimator is usually of order
 $\emph{O}(\frac{1}{nH})$ \citep{Fan1996}. On one hand, similar to the asymptotic
 results in local kernel estimation, the sample size $n$ would become $nH$ in
 the denominator of the dimension-sample-size condition to account for the
 local nature of the estimators. On the other hand, for simplicity, assume that
 $\Omega=[a,b]$ for some constants $a,b\in \mathcal{R}$. To control the
 estimation error or bias for a $p$-dimensional parameter uniformly over $[a,b]$, it is to some degree equivalent to controlling the estimation error of a parameter of dimension proportion to $(b-a)pH^{-1}$. Therefore the numerator in the  dimension-sample-size condition becomes $pH^{-1}$ in our case. 
 
 Note that when the Bayes misclassification rate $R({\bf u})\rightarrow 0$, any classifier with misclassification rate $\hat{R}({\bf u})$ tending to $0$ slower than $R({\bf u})$ would satisfy $|\hat{R}({\bf u})-R({\bf u})|\rightarrow 0$. To better characterize the misclassification rate of our DLPD rule, we establish the following stronger results on the rate of convergence in terms of the ratio $\hat{R}({\bf u})/R({\bf u})$.

 \begin{thm}\label{theorem2}
 	Assume that assumptions (A1)-(A6)
 	and the following assumption hold:
 	\begin{eqnarray}\label{assump2}
 		\sup_{{\bf u}\in \Omega_d}\Delta_p({\bf u})\sup_{{\bf u}\in \Omega_d}\frac{ |\beta({\bf u})|_1}{\Delta_p({\bf u})}=o\left(\left(\frac{n}{\log p}\right)^{\frac{2}{4+d}}\right).
 	\end{eqnarray}
 	For any constant $M>0$, by choosing $\lambda_n=C\left(\frac{\log p}{n}\right)^{\frac{2}{4+d}}\sup_{{\bf u}\in \Omega_d}{\Delta_p({\bf u})}$ for some constant $C$
 	large enough, we have with probability larger than $1-{O}(p^{-M})$,
 	\begin{eqnarray*}
 		\sup_{{\bf u}\in \Omega_d}|\hat{R}({\bf u})/R({\bf u})-1|=
 		O\bigg(\left(\frac{\log p}{n}\right)^{\frac{2}{4+d}}\sup_{{\bf u}\in \Omega_d}\Delta_p({\bf u})\sup_{{\bf u}\in \Omega_d}  \frac{|\beta({\bf u})|_1}{\Delta_p({\bf u})}\bigg).
 	\end{eqnarray*}
 	Consequently, we have
 	\begin{eqnarray*}
 		E_{\bf U}\hat{R}({\bf U})/R-1\rightarrow 0  ~~~{\rm as} ~~p, n\rightarrow \infty.
 	\end{eqnarray*}
 \end{thm}

 \subsection{Minimax lower bound}
 We first introduce the parameter space and some known results in the literature of minimax lower bound theory.  
 We consider the following parameter space:
 \[
 {\cal G}(\kappa)=\left\{(\mu_X({\bf u}),\mu_Y({\bf u}), \Sigma({\bf u})): \mu_X,\mu_Y,\Sigma\in H(2,L),\sup_{{\bf u}\in \Omega_d}\frac{ |\beta({\bf u})|_1^2}{\Delta_p^2({\bf u})}\leq \kappa \right\},
 \]
 where $H(2,L)$ denotes the H\"{o}lder class with order two \citep{Tsybakov}. For definiteness, $\frac{0}{0}$ is defined to be 1. Clearly, assumptions A3 and A6 together imply that $\mu_X({\bf u}),\mu_Y({\bf u})$ and $\Sigma({\bf u})$ belong to the H\"{o}lder class $H(2,L)$ with domain $\Omega_d$. We shall denote $\theta=(\mu_X(\bf u),\mu_Y(\bf u),\Sigma(\bf u))$.
 
 Suppose ${\cal P}$ is a family of probability measures and $\theta $ is the parameter of interest with values in the functional space ${\cal D}$.
 Let $T(\theta)$ be any functional of some parameter $\theta\in {\cal D}$. By noticing that $d(\theta_1,\theta_2):=\sup_{{\bf u}\in \Omega_d}|T(\theta_1)-T(\theta_2)|$  defines a semi-distance for any $\theta_1,\theta_2\in {\cal D}$ , from LeCam's Lemma \citep{LeCam, Yu, Cai2011} we have
 \begin{lem}\label{lecam}
 	Let $T(\theta)$ be any functional of $\theta$ and let $\hat{T}$ be an estimator of $T(\theta)$ on ${\cal P}$ taking values in the metric space $({\cal D},d)$. Let ${\cal D}_0={\theta_0}$ and ${\cal D}_1=\{\theta_1,\ldots,\theta_m\}$ be two $2\delta$-separated subsets of ${\cal D}$ in that $\min_{1\leq i\leq m}d(\theta_0,\theta_i):=\sup_{{\bf u}\in \Omega_d}|T(\theta_0)-T(\theta_i)|>2\delta$. Let $P_i\in {\cal P}$ be the corresponding probability measure for $(\theta_i, {\bf u})$, $i=0,1,\ldots, m$, and let $\bar{P}=\sum_{i=1}^m\omega_iP_{i}$ where $\omega_1,\ldots,\omega_m$ are nonnegative weights such that $\sum_{i=1}^m \omega_i=1$. We then have:
 	\[
 	\inf_{\hat{T}}\sup_{\theta\in{\cal D}}E_\theta\sup_{{\bf u}\in \Omega_d}|\hat{T}(\theta)-T(\theta)| \geq \delta \left(1-\frac{\|\bar{P}-P_0\|_1}{2}\right).
 	\]
 \end{lem}
 By the above version of LeCam' lemma, the derivation of minimax lower bounds thus relies on the construction of the probability measure $P_0$ corresponding to the null hypothesis ${\cal D}_0$, the probability measures $P_1,\ldots, P_m$ corresponding to the alternative ${\cal D}_1$ and the weights $\omega_1,\ldots, \omega_m$ such that (i) $\theta_0,\theta_1,\ldots,\theta_m\in {\cal D}$ and the distance $\min_{1\leq i\leq m}d(\theta_0,\theta_i)$ is as large as possible while (ii) the total variation
 $\frac{1}{2}\|P_0-\bar{P}\|_1$ is controlled to be away from 1. These technical details are deferred to the Appendix. By setting $T(\theta)=\Delta_p(\bf{u})$ and $R({\bf u})$ where $\Delta_p({\bf u})$ and $R({\bf u})$ are defined as in \eqref{R}, the following theorem establishes minimax lower bounds for the Bayes misclassification rate.
 
 \begin{thm}\label{thminimax}
 	Assume that $\kappa=O(p^{\gamma})$ for some constant $0<\gamma<\frac{1}{2}$ and $\kappa=o\left(\left(\frac{n}{\log p}\right)^{\frac{4}{4+d}}\right)$. Let $\tilde{\Delta}_p({\bf u})$ and $\tilde{R}({\bf u})$ be estimators of $\Delta_p({\bf u})$ and $R({\bf u})=\phi\left(-\frac{\Delta_p(\bf u)}{2} \right)$ respectively. Assume that $n_1\asymp n_2$ and let $\alpha=\frac{n(1-2\gamma)}{2en_1}$. We have,
 	\begin{eqnarray}\label{minimax1}
 		~~~~\inf_{\tilde{\Delta}_p}\sup_{\theta\in{\cal G}(\kappa)}E_\theta\sup_{{\bf u}\in \Omega_d}|\tilde{\Delta}_p({\bf u})-\Delta_p({\bf u})| \geq \frac{1}{2}\sqrt{\kappa}\left(\frac{\alpha\log p}{n}\right)^{\frac{2}{4+d}}
 	\end{eqnarray}
 	and
 	\begin{eqnarray}\label{minimax2}
 		~~~~\inf_{\tilde{R}}\sup_{\theta\in{\cal G}(\kappa)}E_\theta\sup_{{\bf u}\in \Omega_d}|\tilde{R}({\bf u})-R({\bf u})| \geq \frac{1}{2}\sqrt{\kappa}\left(\frac{\alpha \log p}{n}\right)^{\frac{2}{4+d}}.
 	\end{eqnarray}
 \end{thm}
 
 Note that the upper bound we have obtained in Theorem \ref{theorem1} is of order $\sqrt{\kappa}\left(\frac{\log p}{n}\right)^{\frac{2}{4+d}}$ in ${\cal G}_{\kappa}$. Together with Theorem \ref{thminimax} we conclude that the misclassification rate of our proposed DLPD achieves the optimal rate of convergence over ${\cal G}_{\kappa}$.  Moreover,  since the lower bound in Theorem \ref{thminimax} is not negligible when $\sqrt{\kappa}$ has the same order as $\left(\frac{n}{\log p}\right)^{\frac{2}{4+d}}$ while \eqref{mm3} is negligible when $\kappa=O(p^{\gamma})$, we conclude that the detection boundary \eqref{assump0} for $\sup_{{\bf u}\in \Omega_d}\frac{ |\beta({\bf u})|_1}{\Delta_p({\bf u})}$ is optimal when $\left(\frac{n}{\log p}\right)^{\frac{4}{4+d}} =O(p^{\gamma})$ where $\gamma \in (0,1/2)$.  
 \section{Numerical studies}
 
 \subsection{Choice of tuning parameters}
 The bandwidths for the mean functions $\hat{\mu}_X({\bf u})$ are chosen using the classical leave-one-out cross validation. Once we obtain the bandwidth for estimating ${\mu}_X({\bf u})$, the bandwidth matrix for the covariance functions $\hat{\Sigma}_X({\bf u})$ can be obtained using a similar leave-on-out procedure. More specifically, for $i=1,\ldots, n_1$, we denote the estimators of $\Sigma({\bf U}_i)$ obtained by leaving the $i$th sample out as $\hat{\mu}_{X,-i}({\bf U}_i)$ and let $\hat{\Sigma}_{X,-i}({\bf U}_i) $ be the mean function estimator with the bandwidth chosen by leave-one-out cross validation. We then choose $H_x$ such that
 \begin{eqnarray}\label{cvh}
 	r_{cv}(H_x)=\frac{1}{p^2n_1}\sum_{i=1}^{n_1}\left\| \big(X_{i}-\hat{\mu}_{X,-i}({\bf U}_i)\big) \big(X_{i}-\hat{\mu}_{X,-i}({\bf U}_i)\big)^T-\hat{\Sigma}_{X,-i}({\bf U}_i) 
 	\right\|_F^2,
 \end{eqnarray}
 is minimized. Here $\|\cdot\|_F$ denotes the matrix Frobenius norm. The bandwidths  for computing $\hat{\mu}_Y({\bf u})$ and $\hat{\Sigma}_Y({\bf u})$ are chosen similarly. 
 
 %
 Following \cite{Tsybakov}, we define the weighted  Mean Integrated Squared Error (MISE) as:
 \[
 r(H_x)=p^{-2}E\int_{\Omega_d}\|\hat{\Sigma}_{X,-i}({\bf u})-\Sigma({\bf u}) \|^2_F f(\bf u)d{\bf u}.
 \]
 The following theorem indicates
 that  the cross-validation  criterion  $r_{cv}(H_x)$ in meaningful in the sense that it provides an  estimator for the weighted  MISE $ r(H_x)$ subject to a constant shift (independent of $H_x$), and a negligible bias. 
 \begin{thm}\label{misethm}
 	Under assumptions (A1)-(A6), we have,
 	\[	Er_{cv}(H_x)= E  r(H_x)+C_\sigma+O\left(\Big(\frac{\log p}{n}\Big)^{\frac{2}{2+d}}\right),
 	\]
 	where $C_\sigma=p^{-2}	E\big\| \big(X_{i}-{\mu}_X({\bf U}_i)\big) \big(X_{i}-{\mu}_X({\bf U}_i)\big)^T- \Sigma({\bf U}_i)\big\|_F^2$ is a constant shift.
 \end{thm}
 The proof of the above theorem is provided in the Appendix.
 Now we obtain the bandwidths for computing the estimators $\hat{\Sigma}_X({\bf u})$, $\hat{\mu}_X({\bf u})$, $\hat{\Sigma}_Y({\bf u})$ and $\hat{\mu}_Y({\bf u})$. For a given $\lambda_n$, the convex
 optimization problem (\ref{hatbeta}) is implemented via linear programming as
 \begin{eqnarray*}
 	&&\min \sum_{i=1}^{p}v_i~~~ {\rm subject~ to}~~-v_i\leq \beta_i \leq
 	v_i\\
 	&&~{\rm and}-\lambda_n\leq \gamma_i^T({\bf u})\beta -(\hat{\mu}_X^i({\bf u})-\hat{\mu}_Y^i({\bf u}))\leq
 	\lambda_n,~~i=1,\ldots,p,
 \end{eqnarray*}
 where $v=(v_1,\ldots,v_p)^T\in \mathbf{R}^{p} $ and $\gamma_i({\bf u})^T$ is the
 $i$-th row of $\hat{\Sigma}({\bf u})$.
 
 This is similar to  the Dantzig selector \citep{candes,Cai}. The tuning parameter
 $\lambda_n$ in (\ref{hatbeta}) is chosen using $K$-fold cross
 validation. More specifically, randomly divide the index set
 $\{1,\ldots,n_1\}$ into $K$ subgroups $N_{11},\ldots,$ $ N_{1K}$, and
 divide $\{1,\ldots,n_2\}$ into $K$ subgroups $N_{21},\ldots,N_{2K}$.
 Denote the full sample set as $S=\{(X_i,{\bf U}_i),(Y_j,{\bf V}_j): 1\leq i\leq n_1, 1\leq
 j\leq n_2\}$ and let $S_k=\{(X_i,{\bf U}_i),(Y_j,{\bf V}_j): i\in N_{1k},Y\in N_{2k}\}$ for
 $k=1,\ldots, K$. For a given $\lambda_n$ and  $1\leq k\leq K$, let $\hat{\mu}^{(k)}_X({\bf u})$, $\hat{\mu}^{(k)}_Y({\bf u})$ and $\hat{\beta}^{(k)}({\bf u})$ be estimators of $\mu_X({\bf u})$, $\mu_Y({\bf u})$ and $\beta({\bf u})$ computed using (\ref{muhatx}), (\ref{muhaty}) and (\ref{hatbeta}), samples in $S\setminus S_k$ and bandwidths $H_x, H_y$. For each
 $k=1,\ldots, K$, let
 \[ C_{1k}=\sum_{i\in
 	N_{1k}}I_{\{ [X_i-(\hat{\mu}^{(k)}_X({\bf U}_i)-\hat{\mu}^{(k)}_Y({\bf U}_i))/2]^T\hat{\beta}^{(k)}({\bf U}_i)\geq 0\}},\]
 and
 \[C_{2k}=\sum_{i\in
 	N_{2k}}I_{\{ [Y_i-(\hat{\mu}^{(k)}_X({\bf V}_i)-\hat{\mu}^{(k)}_Y({\bf V}_i))/2]^T\hat{\beta}^{(k)}({\bf V}_i)\leq 0\}}.\]
 Here $I_{\{\cdot\}}$ is the indicator function. Clearly,
 $C_{1k}+C_{2k}$ gives the total number of correct classification for
 the test data set $S_k$ using the DLPD rule based on $S\setminus
 S_k$. We then find $\lambda_n$ such that the following averaged
 correct classification number is maximized:
 \begin{eqnarray*}
 	CV(\lambda_n)=\frac{1}{K}\sum_{k=1}^K(C_{1k}+C_{2k}).
 \end{eqnarray*}
 We remark that 
 local smoothing estimates are obtained in our method before applying linear programming. Hence the computation time consists of the time for local smoothing and the time for linear programming.  The proposed method is computationally manageable for large dimensional data. 
 
 
 To speed up computation, instead of fitting the classifier for every new observation, we may fit it on a sufficient fine grid of ${\bf u}$ and interpolate when a new instance comes. Here we provide an argument when the dynamic factor ${\bf u}$ is one-dimensional on an interval denoted as $\Omega=[a,b]$.  Assume that  $\Delta({\bf u})$ has continuous first derivative on $\Omega$. Suppose the classifier is fitted on the grid of points denoted as  $u_i=a+(i-1)(b-a)/k$ for $i=1,\ldots, k+1$. For any $U_Z\in [a,b]$, we can simply use the classifier fitted in the nearest point, say $u_t$ with $t\in\{1,\ldots, k+1\}$, for classifying the new observation with ${\bf u}=U_Z$. In particular,  by choosing  $k=O\left( \left(\frac{\log p}{n}\right)^{\frac{4+d}{2}}\left(\sup_{{\bf u}\in \Omega_d}\frac{|\beta({\bf u})|_1}{\Delta_p({\bf u})}\right)^{-1}\right)$, we can show that the conditional misclassification rate $\hat{R}(U_Z)$ of this interpolated classifier satisfies $\hat{R}(U_Z)-R(U_Z)\leq \hat{R}(U_Z)-R(u_t)+|R(u_t)-R(U_Z)|=O\left( \left(\frac{\log p}{n}\right)^{\frac{2}{4+d}}\sup_{{\bf u}\in \Omega_d}\frac{|\beta({\bf u})|_1}{\Delta_p({\bf u})}\right)$ under the assumptions of Theorem 3.1. This implies that the order of the error rate remain unchanged when $k$ is large enough.  
 
 \subsection{Simulation}
 For the simulation study, we consider the following four models:
 
 {\rm Model 1.} We generate $U_1,\ldots, U_{n_1},V_{1},\ldots, V_{n2}$ independently from $U[0,1]$, and generate $X_i\sim N(\mu_X(U_i),\Sigma(U_i))$, $Y_j\sim N(\mu_Y(V_j),\Sigma(V_j))$ for $1\leq i\leq n_1, 1\leq j\leq n_2$. The mean functions $\mu_X(u)=(\mu^1_X(u),\ldots,\mu_X^p(u))^T$ and $\mu_Y(v)=(\mu^1_Y(v),\ldots,\mu_Y^p(v))^T$ are set as $\mu^1_X(u)=\cdots=\mu^p_X(u)=1$, $\mu^1_Y(v)=\cdots=\mu^{20}_Y(v)=0$ and $\mu^{21}_Y(v)=\cdots=\mu^p_Y(v)=1$. The covariance matrix is set as $\Sigma(u)=(0.5^{|i-j|})_{1\leq i,j\leq p}$.
 
 {\rm Model 2.} We generate $U_1,\ldots, U_{n_1},V_{1},\ldots, V_{n2}$ independently from $U[0,1]$, and generate $X_i\sim N(\mu_X(U_i),\Sigma(U_i))$, $Y_j\sim N(\mu_Y(V_j),\Sigma(V_j))$ for $1\leq i\leq n_1, 1\leq j\leq n_2$. The mean functions $\mu_X(u)=(\mu^1_X(u),\ldots,\mu_X^p(u))^T$ and $\mu_Y(v)=(\mu^1_Y(v),\ldots,\mu_Y^p(v))^T$ are set to be $\mu^1_X(u)=\cdots=\mu^p_X(u)=\exp(u)$, $\mu^1_Y(v)=\cdots=\mu^{20}_Y(v)=v$ and $\mu^{21}_Y(v)=\cdots=\mu^p_Y(v)=\exp(v)$. The covariance matrix is set as $\Sigma(u)=(u^{|i-j|})_{1\leq i,j\leq p}$.
 
 {\rm Model 3.} We take the same model as Model 2 except that the mean functions are set to be $\mu^1_X(u)=\cdots=\mu^p_X(u)=u$, $\mu^1_Y(v)=\cdots=\mu^{20}_Y(v)=-v$ and $\mu^{21}_Y(v)=\cdots=\mu^p_Y(v)=v$, and the covariance matrix is set to be $\Sigma(u)=(u)_{1\leq i,j\leq p}+(1-u)I_{p}$.

 {\rm Model 4.} We take $d=2$ and let ${\bf U}_1=(U^{(1)}_1,U^{(2)}_1), \ldots, {\bf U}_{n_1}=(U^{(1)}_{n_1},U^{(2)}_{n_1})$ and ${\bf V}_1=(V^{(1)}_1,V^{(2)}_1), \ldots, {\bf V}_{n_2}=(V^{(1)}_{n_2},V^{(2)}_{n_2})$. We generate $U_{i}^{(k)}, V_{j}^{(k)}$ independently from $U[0,1]$ for $1\leq i\leq n_1,  1\leq j \leq n_2, k=1,2$.  We then generate $X_i\sim N(\mu_X({\bf U}_i),\Sigma({\bf U}_i))$, $Y_j\sim N(\mu_Y({\bf V}_j),\Sigma({\bf V}_j))$ for $1\leq i\leq n_1, 1\leq j\leq n_2$, where
 $\mu^1_X({\bf U})=\cdots=\mu^{20}_X({\bf U})=0.5+\sin(U^{(1)}+U^{(2)})$,  $\mu^{21}_X({\bf U})=\cdots=\mu^{p}_X({\bf U})=\cos(U^{(1)}+U^{(2)})$,  $\mu^1_Y({\bf V})=\cdots=\mu^{p}_Y({\bf V})=\cos(V^{(1)}+V^{(2)})$ and $\Sigma({\bf U})=\Big(\frac{|U^{(1)}-U^{(2)}|}{U^{(1)}+U^{(2)}}\Big)_{1\leq i,j\leq p}+\Big(1-\frac{|U^{(1)}-U^{(2)}|}{U^{(1)}+U^{(2)}}\Big)I_{p}$.
 
 \begin{figure}[htbp]
 	\centering
 	\includegraphics[width=0.8\textwidth]{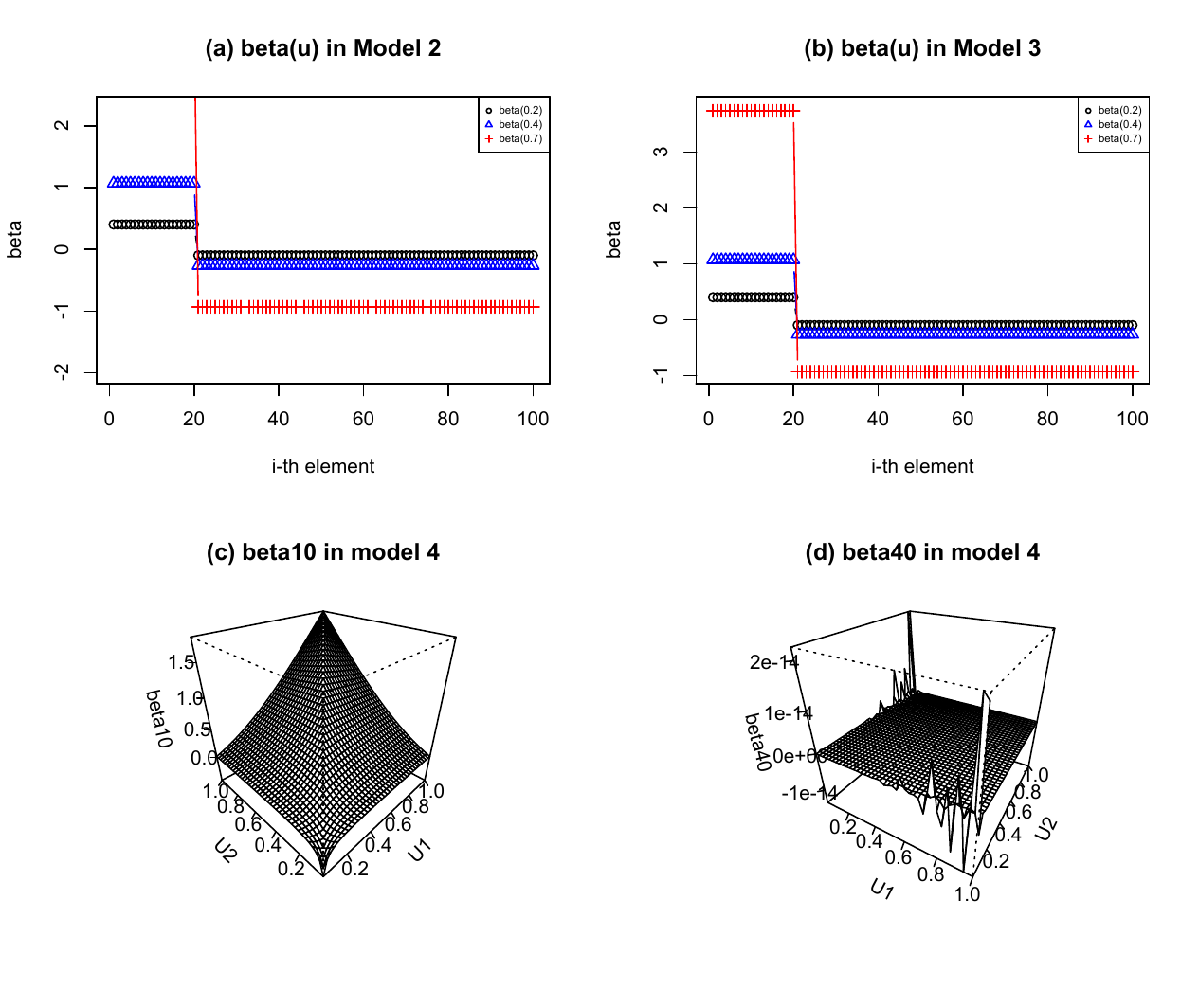}
 	\caption{$\beta({\bf U})$ in Models 2-4 when $p=100$: (a) plot of $\beta(u)$ for $u=0.2, 0.4, 0.7$ under Model 2; (b) plot of $\beta(u)$ for $u=0.2, 0.4, 0.7$ under Model 3; (c) $\beta_{10}(U_1,U_2)$ under Model 4; (d) $\beta_{40}(U_1,U_2)$ under Model 4.}\label{fig2}
 \end{figure}

 Model 1 is a static case where the means and the covariances are independent of the covariate. The other three models are dynamic ones. Under Models 1-4, $\beta(u)$ is approximately sparse in the sense that some of the elements of $|\beta_i(u)|$'s have large values while others are much smaller. Figure 2 (a) and (b) show $\beta(u)$ for $u=0.2,0.4,0.7$ under Models 2 and 3 when $p=100$. Generally, under Model 2, $\beta_1(u),\ldots,\beta_{21}(u)$ are nonzero while $\beta_{22}(u),\ldots,\beta_p(u)$ are very close to zero. Under Model 3, $\beta_1(u),\ldots,\beta_{21}(u)$ are much larger than $\beta_{22}(u),\ldots,\beta_p(u)$, which are not necessarily close to zero. Figure 2 (c) and (d) plot $\beta_{10}(U_1,U_2)$ and $\beta_{40}(U_1,U_2)$ as functions of $U_1,U_2$ under Model 4 when $p=100$. Clearly, $\beta_{1}(U_1,U_2),\ldots, \beta_{21}(U_1,U_2)$ have various shapes as functions of $U_1, U_2$ and $\beta_{22}(U_1,U_2),\ldots, \beta_{p}(U_1,U_2)$ are very close to zero.
 

 For each model we consider $p=50,100, 200$ and $n_1=n_2=100$. We generate 100 samples from population $(X,{\bf U})$ and 100 samples from population $(Y,{\bf V})$ as testing samples to compute the misclassification rate $R_{\rm dlpd}$ of our DLPD rule. Gaussian kernel function is used in our DLPD rule. For comparison, we also use the LPD rule in \cite{Cai}, the support vector machine (SVM) with a linear kernel, and the k-nearest-neighbor (KNN) algorithm to classify these 200 testing samples and compute their misclassification rates, denoted as $R_{\rm lpd}$, $R_{\rm svm}$ $R_{\rm knn}$ respectively.  The  $k$ in KNN is chosen using a bootstrapping algorithm in \cite{ph}.   The optimal Bayes risk is denoted by $R$. The procedure is repeated for 100 times. The mean and standard deviation of the misclassification rates over these 100 replications are reported in Table \ref{sim1}.  
 From Table \ref{sim1}, we can see that the $R_{slpd}$ values are all very close to the optimal misclassification rate $R$, and are relatively smaller than the mean misclassification rates of other methods. Overall, the numerical performance of DLPD is better than other methods. Although Model 1 favors the LPD method, we observe that our DLPD rule works as well as the LPD rule. Interestingly, the linear LPD approach is also performing well in all the cases, implying that the linear method is robust in some sense. From the formulation of the Nadaraya-Watson estimators introduced in Section 2, we know that loosely speaking, LPD can be viewed as a special case of DLPD when the bandwidths tend to infinity. Therefore, practically we would expect DLPD to outperform LPD under dynamic assumptions and work as well as LPD under static assumptions as long as the bandwidths in the numerical study is taken to be large enough.

 \begin{table}
 	\centering
 	\begin{tabular}{|cccc|ccc|}
 		\hline
 		p                                   &       50    &   100   &   200 &       50    &   100   &   200  \\
 		\hline
 		& \multicolumn{3}{c|}{Model 1}  & \multicolumn{3}{c|}{Model 2}  \\
 		$R$                       &   0.083 (-)              & 0.083 (-)     &  0.083 (-)  &   0.041 (-)         & 0.041 (-)     &  0.041 (-)   \\
 		$R_{\rm dlpd}$      &  0.104(0.023)      &  0.110 (0.023)   &  0.111 (0.022)    & 0.086 (0.020)  & 0.102 (0.024)   & 0.108 (0.024) \\
 		$R_{\rm lpd}$         &  0.103 (0.023)       & 0.111 (0.021)    &  0.113 (0.025)  & 0.113 (0.025)  &   0.116 (0.021) & 0.115 (0.027) \\
 		$R_{\rm svm}$               & 0.152(0.033)  & 0.157 (0.027)   & 0.160 (0.028)    & 0.159 (0.039)  & 0.161 (0.039)  & 0.162 (0.033)\\
 		$R_{\rm knn}$               & 0.143(0.029)  & 0.172 (0.036)  & 0.225 (0.038)   & 0.155 (0.038)  & 0.178 (0.052)   & 0.210 (0.062)\\
 		\hline
 		& \multicolumn{3}{c|}{Model 3}  & \multicolumn{3}{c|}{Model 4}  \\
 		$R$                       &   0.092 (-)  & 0.083 (-)   &   0.079 (-)    &  0.095 (-)  &  0.084 (-) &   0.079 (-)\\
 		$R_{\rm dlpd}$       &   0.145 (0.025)    &  0.141 (0.025)    &  0.143 (0.027)     & 0.191 (0.032) & 0.189 (0.030)  &  0.187 (0.033) \\
 		$R_{\rm lpd}$           &   0.162 (0.025)     &  0.153 (0.026)   &  0.154 (0.027)  & 0.199 (0.036)  & 0.194 (0.033) & 0.197 (0.032)  \\
 		$R_{\rm svm}$       &   0.161 (0.031) & 0.148 (0.028)  & 0.137 (0.024)   &  0.227 (0.039)   & 0.226 (0.036) & 0.217 (0.042)\\
 		$R_{\rm knn}$        &  0.194 (0.026) & 0.217 (0.029)   & 0.228 (0.034)  &  0.283 (0.054)  & 0.333 (0.062) & 0.386 (0.053)\\
 		\hline
 	\end{tabular}
 	\caption{The misclassification rates of DLPD, LPD, SVM, KNN, and the optimal misclassification rate R under Models 1-4. }\label{sim1}
 \end{table}
 
 \subsection{Breast cancer study}
 Breast cancer is the second leading cause of deaths from cancer among women in the United States. Despite major progresses in breast cancer treatment, the ability to predict the metastatic behavior of tumor remains limited. This breast cancer study was first reported in \cite{van2002gene} where 97 lymph node-negative breast cancer patients, 55 years old or younger, participated in this study. Among them, 46 developed distant metastases within 5 years ($X$ class) and 51 remained metastases free for at least 5 years ($Y$ class). In this study, covariates including clinical risk factors (tumor size, age, histological grade etc.)  as well as expression levels for 24,481 gene probes were collected. The histograms of the tumor sizes for both classes are presented in Figure \ref{figts}. Shapiro's normality test is used to test the normality of the tumor size with $p$-value $<0.001$ for class $X$ and $0.221$ for class $Y$, indicating that it might not be suitable to treat tumor size as one of the covariates to conduct classification using the LPD rule.
 On the other hand, as introduced before, Figure \ref{figgene} indicates that the gene expression levels for patients in the $X$ class and the $Y$ class vary differently as tumor size changes. We thus set the tumor size as the dynamic factor. For comparison, we also consider the LPD rule with or without including the tumor size as one of the covariates, denoted as ``LPD with $U$" and ``LPD without $U$", respectively. The intercept is chosen according to Proposition 2 of \cite{Maiqing2012}. For simplicity, we use the $p$ genes with the largest absolute $t$-statistic values between the two groups for discriminant analysis, and in our study we set $p=25, 50, 100$ and $200$. We randomly choose 92 observations as training samples and set the rest 5 observations as test samples. This procedure is repeated for 100 times. The mean misclassification rate and its standard deviation over 100 replications are reported in Table \ref{bc}. From the results we can see that no significant improvement is observed when the tumor size is included as one of the covariate in the LPD rule. However, when it is set to be a dynamic factor as in our DLPD rule, the misclassification rate of is seen to be reduced. 
 \begin{figure}[htbp]
 	\centering
 	\includegraphics[width=0.8\textwidth]{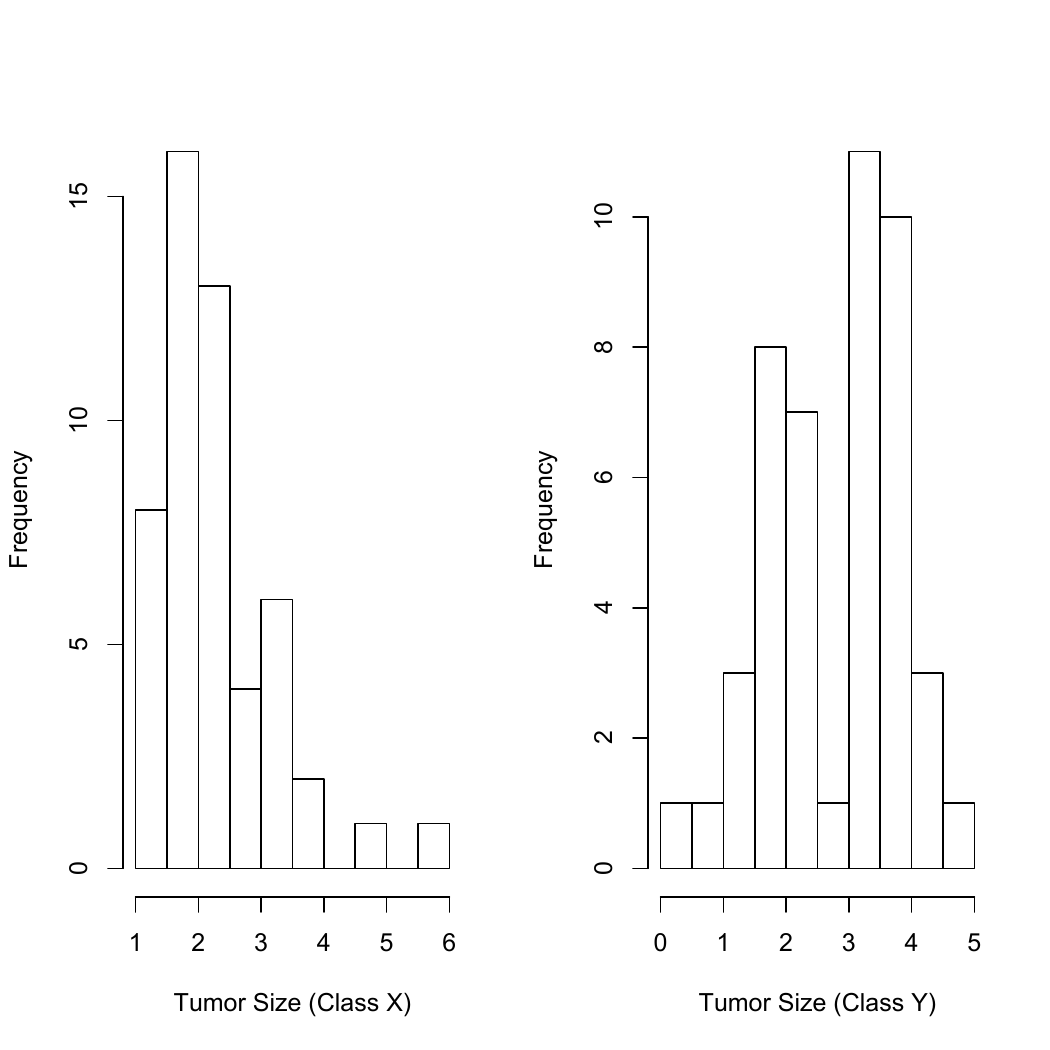}
 	\caption{Histogram of  tumor sizes in the $X$ and $Y$ classes of the Breast Cancer data.}\label{figts}
 \end{figure}
 
 %
 
 \begin{table}
 	\centering
 	\begin{tabular}{|cccc|}
 		\hline
 		p	                 &     	LPD without $U$   &  LPD with $U$ & DLPD     \\
 		\hline
 		25                        & 0.198 (0.018)  & 0.198 (0.018)   & 0.190 (0.018)  \\
 		50                       & 0.184 (0.016)  & 0.184 (0.016)  & 0.170 (0.016)  \\
 		100                       & 0.186 (0.015)  & 0.184 (0.016)   & 0.172 (0.015)   \\
 		200                       & 0.216 (0.019)  & 0.216 (0.020)   & 0.204 (0.019)   \\
 		\hline
 	\end{tabular}
 	\caption[]{ Mean classification rate and its standard deviation for the Breast Cancer study over 100 replications. }\label{bc}
 \end{table}

 \section{Conclusion and discussion}
 We have proposed a new and simple model for high dimensional linear discriminant analysis when data is high-dimensional and the local features of the data can play important roles in classification.
 Our approach combines the simplicity of kernel smoothing and the powerful method of regularization for studying high dimensional problems. We have established uniform Bernstein-type inequalities for our high-dimensional nonparametric estimators, and shown that the risk function of the proposed DLPD rule converges to the optimal Bayes risk in probability under general sparsity assumptions, uniformly over the range of the covariates. The minimax lower bounds for the estimation of the Bayes risk are also established, and it is shown that the misclassification rate of our proposed rule is minimax-rate optimal. Both the uniform convergence  and the minimax results appear to be new in the literature of classification.
 
 A limitation of the linear discriminant rule is its Gaussian assumption. An
 immediate generalization of our method is to allow a more flexible family of
 distributions, for example, the transnormal family in
 \cite{Lin}. On the other hand, the smoothness assumption (A6) might not be appropriate in some cases. For example, discontinuity of the brain activity is common in certain applications \citep{Vairavan2009},  leading to discontinuous $\mu({\bf u})$ and $\Sigma({\bf u})$ which are usually modeled as piecewise-continuous functions. This gives rise to a similar problem as ours where the aim is to identify the number of the discontinuous points and their locations.
 We also remark that the assumption of independent observations $(X_i,{\bf U}_i), i=1,\ldots, n_1, (Y_j,{\bf V}_j), j=1,\ldots, n_2$ can be relaxed to that $X_i| {\bf U}_i,i=1,\ldots,n_1, Y_j| {\bf V}_j, j=1,\ldots, n_2$ are 
 weakly dependent, which might enable us to incorporate temporal correlations. Under suitable weakly dependence assumptions such as strongly mixing \citep{Merl},  
 estimators of the components in the Bayes rule proposed in Section 2 can be shown consistent. Nevertheless, for time series data, it would be interesting to incorporate our DLPD rule with time series models so as to capture the structures of the covariance matrix and the dependency among the sequences of observations. This is beyond the scope of the current paper and will be studied elsewhere.

 In our work, we have assumed that $\Sigma_X({\bf u})=\Sigma_Y({\bf u})$ which seems
 to be reasonable for the data analysis. It is however worth considering
 problems where covariances are dynamic but not equal.  Finally, we have only discussed binary
 classification in this paper. It will be interesting to extend this work to study 
 multiclass classification \citep{Pan:etal:2015,Mai:etal:2015},  and other recent approaches which considered more complex structures  \citep{lr}. Last but not least, in the unbalanced case, the cut-off point in the Bayes procedure becomes $\log(\pi_2/\pi_1)$, which is usually estimated by $\log(n_2/n_1)$. Here $\pi_1$ is the prior probability of observing a sample from Class $X$ and $\pi_2=1-\pi_1$. However, as pointed out in \cite{Maiqing2012}, the problem of finding the right cut-off point receives little attention in the literature and it is also important to find a optimal estimator of the cut-off points to improve classification accuracy.
 
 One alternative of our DLPD rule is to develop a dynamic logistic regression model in which a rule is obtained by minimizing a dynamic version of the penalized entropy loss. It is well know that under Gaussian assumptions, logistic regression and LDA are equivalent in that the solution (in a population sense) of logistic regression is exactly the Bayes rule. For the fixed dimension and static case, earlier numerical studies have shown that logistic regression and LDA would give the same linear discriminant function \citep{Press}, while theoretically, \cite{efron1975} showed that LDA is more efficient than logistic regression under the Gaussian assumptions. On one hand, it is worth exploring the theoretical properties of logistic regression based rules under the more general sparsity assumption (3.1). On the other hand, it would be interesting and challenging to compare the efficiency of logistic regression rules and LDA rules under both high dimensional and non-stationary assumptions.
 
 As pointed out in \cite{candes}, a two-stage procedure generally produces better estimation results in the strict sparse case where many parameters are zero. When the true discriminant direction $\beta(u)$ is sparse, we may use a two-stage procedure similar to the one in \cite{spl}. That is, in the first stage, the same bandwidth is used to obtain a sparse first stage estimator $\hat{\beta}(u)$. In the second stage, we can apply our approach again to estimate the identified nonzero elements in $\hat{\beta}(u)$. If the number of nonzeros in the second stage is very low, different bandwidths can be considered for different elements.  
 
 \appendix
 
 \section*{Appendix A}

 Before we proceed to the proofs for the main theorems, we introduce some technical lemmas.

 \begin{lem}\label{kernel}
 	Suppose $\epsilon_n\rightarrow 0$, $n|H_x|\epsilon_n^2\rightarrow \infty$ and there exists a large enough constant $C_h$ such that $\epsilon^2_n>C_h(h_x^4+h_y^4)$. Under assumptions (A1)-(A6), there exist constants $C_1>0$ and $C_2>0$ such that
 	\begin{eqnarray*}
 		&&P\Big(\sup_{{\bf u}\in \Omega_d}\Big|\frac{1}{n_1}\sum_{i=1}^{n_1}{\bf K}_{H_x}({\bf U}_i-{\bf u})-f({\bf u})\Big|\geq \epsilon_n\Big)\\
 		&\leq& C_1\left(\frac{n}{\log p}\right)^{\frac{4}{4+d}}\exp\left\{-C_2n \left(\frac{\log p}{n}\right)^{\frac{d}{4+d}}\epsilon_n^2\right\};
 	\end{eqnarray*}
 	and
 	\begin{eqnarray*}
 		&&P\Big(\sup_{{\bf u}\in \Omega_d}\Big|\frac{1}{n_2}\sum_{i=1}^{n_2}{\bf K}_{H_y}({\bf V}_i-{\bf u})-f({\bf u})\Big|\geq \epsilon_n\Big)\\
 		&\leq& C_1\left(\frac{n}{\log p}\right)^{\frac{4}{4+d}}\exp\left\{-C_2n \left(\frac{\log p}{n}\right)^{\frac{d}{4+d}}\epsilon_n^2\right\}.
 	\end{eqnarray*}
 \end{lem}
 \begin{proof}
 	
 	Without loss of generality, assume that $\Omega_d=[a_1,b_1]\times\cdots[a_d,b_d]$ and decompose it as $\Omega_d=\cup_{1\leq i_j \leq q_j, j=1,\ldots,d} \omega_{i_1,\ldots,i_d}$, where $q_j=\frac{b_j-a_j}{2h_{xj}^4}$ for $j=1,\ldots, d$ and $\omega_{i_1,\ldots,i_d}=[a_1+2(i_1-1)h_{x1}^4, a_1+2i_1h_{x1}^4]\times\cdots\times[a_d+2(i_d-1)h_{xd}^4, a_d+2i_dh_{xd}^4]$. Denote $W_i({\bf u})=[{\bf K}_{H_x}({\bf U}_i-{\bf u})-E{\bf K}_{H_x}({\bf U}_i-{\bf u})]/n_1$.  We then have:
 	for any ${\bf u}\in \Omega_d$, using assumption (A1) and Markov's inequality we have, for any $0<t<\frac{n_1|H_x|}{2K_1}$,
 	\begin{eqnarray}\label{lamma1.1}
 		&&P\left(\Big|\frac{1}{n_1}\sum_{j=1}^{n_1}{\bf K}_{H_x}({\bf U}_j-{\bf u})-E{\bf K}_{H_x}({\bf U}_i-{\bf u})\Big|>\epsilon_n \right) \\
 		&\leq&  2 \exp\{-t\epsilon_n\}\Pi_{i=1}^{n_1}E  \exp\{tW_i({\bf u})\}       \nonumber\\
 		&\leq& 2 \exp\{-t\epsilon_n\}\Pi_{i=1}^{n_1}\{1+t^2EW_i({\bf u})^2 \}  \nonumber\\
 		&\leq& 2\exp\left\{-t\epsilon_n+\sum_{i=1}^{n_1}t^2E W_i({\bf u})^2\right\} \nonumber\\
 		&\leq& 2\exp\left\{-t\epsilon_n+\frac{Ct^2}{n_1|H_x|}\right\}, \nonumber
 	\end{eqnarray}
 	for some large enough constant $C$. By setting $t=(2C)^{-1}n_1|H_x|\epsilon_n$, we have:
 	\begin{equation}\label{lemma1.2}
 		P\left(\Big|\frac{1}{n_1}\sum_{j=1}^{n_1}{\bf K}_{H_x}({\bf U}_j-{\bf u})-E{\bf K}_{H_x}({\bf U}_i-{\bf u})\Big|>\epsilon_n \right) \leq 2\exp\left\{-\frac{\epsilon_n^2n_1|H_x|}{4C}\right\}.
 	\end{equation}
 	Write ${\bf u}_{i_1,\ldots,i_d}=(a_1+2i_1h_{x1}^4,\cdots, a_d+2i_dh_{xd}^4)^T$ for $1\leq i_j \leq q_j, j=1,\ldots,d$. Note that
 	\begin{align}\label{lemma1.3}
 		&\sup_{{\bf u}\in \Omega_d}\Big|\frac{1}{n_1}\sum_{j=1}^{n_1}{\bf K}_{H_x}({\bf U}_j-{\bf u})-E{\bf K}_{H_x}({\bf U}_i-{\bf u})\Big| \\
 		\leq&
 		\max_{1\leq i_j \leq q_j, j=1,\ldots,d}\Big|\frac{1}{n_1}\sum_{j=1}^{n_1}{\bf K}_{H_x}({\bf U}_j-{\bf u}_{i_1,\ldots,i_d})-E{\bf K}_{H_x}({\bf U}_i-{\bf u}_{i_1,\ldots,i_d})\Big|+ \nonumber\\
 		&\max_{1\leq i_j \leq q_j, j=1,\ldots,d}\sup_{{\bf u}\in\omega_{i_1,\ldots,i_d}}\Big|\frac{1}{n_1}\sum_{j=1}^{n_1}{\bf K}_{H_x}({\bf U}_j-{\bf u})-\frac{1}{n_1}\sum_{j=1}^{n_1}{\bf K}_{H_x}({\bf U}_j- {\bf u}_{i_1,\ldots,i_d})  \nonumber\\
 		&-[E{\bf K}_{H_x}({\bf U}_i-{\bf u})-E{\bf K}_{H_x}({\bf U}_i-{\bf u}_{i_1,\ldots,i_d})] \Big|. \nonumber
 	\end{align}
 	Denote ${\rm diag}(H^4_x)=(h^4_{x1},\ldots,h_{xd}^4)^T$. For each $(i_1,\ldots,i_d)$, using the mean value theorem and assumption (A1) we have, there exist random scalars $0\leq R_{i_1,\ldots,i_d}\leq 2$ depending on ${\bf U}_1,\ldots,{\bf U}_{n_1}$ such that
 	\begin{align}\label{lemma1.4}
 		&\sup_{{\bf u}\in\omega_{i_1,\ldots,i_d}}\Big|\frac{1}{n_1}\sum_{j=1}^{n_1}{\bf K}_{H_x}({\bf U}_j-{\bf u})-\frac{1}{n_1}\sum_{j=1}^{n_1}{\bf K}_{H_x}({\bf U}_j-{\bf u}_{i_1,\ldots,i_d}) \Big|   \\
 		\leq& \frac{2{\rm tr}(H^3_x)}{n_1}\sum_{j=1}^{n_1} \sup_{u\in\omega_{i_1,\ldots,i_d}}\Big|{\bf K}'_{H_x}({\bf U}_j-{\bf u}_{i_1,\ldots,i_d}+R_{i_1,\ldots,i_d}{\rm diag}(H_x^4))  \Big|_\infty \nonumber\\
 		\leq& 2K_2{\rm tr}(H^3_x). \nonumber
 	\end{align}
 	On the other hand, it can be easily shown that
 	\begin{equation}\label{lemma1.5}
 		\max_{1\leq i_j \leq q_j, j=1,\ldots,d}\sup_{{\bf u}\in\omega_{i_1,\ldots,i_d}}[E{\bf K}_{H_x}({\bf U}_i-{\bf u})-E{\bf K}_{H_x}({\bf U}_i-{\bf u}_{i_1,\ldots,i_d})] =\emph{O}({\rm tr}(H^4_x)).
 	\end{equation}
 	Combining (\ref{lemma1.2}), (\ref{lemma1.3}), (\ref{lemma1.4}) and (\ref{lemma1.5}) with the assumption on $\epsilon_n$, we have
 	\begin{align}\label{lemma1.6}
 		&P\left(\sup_{{\bf u}\in \Omega_d}\Big|\frac{1}{n_1}\sum_{j=1}^{n_1}{\bf K}_{H_x}({\bf U}_j-{\bf u})-E{\bf K}_{H_x}({\bf U}_i-{\bf u})\Big|>\epsilon_n \right)  \\
 		&\leq c_1q_1\cdots q_d \exp\{-c_2n_1|H_x|\epsilon_n^2\}  \nonumber\\
 		&=C_1\left(\frac{n_1}{\log p}\right)^{4/(4+d)}\exp\{-C_2n_1h_x^{d}\epsilon_n^2\}, \nonumber
 	\end{align}
 	for  some constants $c_1>0, c_2>0, C_1>0, C_2>0$. 
 	The first argument of Lemma \ref{kernel} is then proved by combining (\ref{lemma1.6}) and the following well known result (see for example \cite{Pagan}):
 	\begin{eqnarray*}
 		\sup_{u\in \Omega_d}|E{\bf K}_{H_x}({\bf U}_i-{\bf u})-f({\bf u})|=\emph{O}({\rm tr}H^2).
 	\end{eqnarray*}
 	The second argument of Lemma \ref{kernel}  can be proved similarly.
 	
 \end{proof}
 
 Lemma \ref{mean} and Lemma \ref{covariance} below give the Bernstein-type inequalities (uniformly in ${\bf u}\in \Omega_d$) for the functional estimators of the means and covariance matrix defined as in (\ref{muhatx}), (\ref{muhaty}) and (\ref{sigmahat}).  We only provide the proof for Lemma \ref{covariance} and the proof for Lemma \ref{mean} is similar.
 
 \begin{lem}\label{mean}
 	Suppose $\epsilon_n\rightarrow 0$, $n|H_x|\epsilon_n^2\rightarrow \infty$ and there exists a large enough constant $C_h$ such that $\epsilon^2_n>C_h(h_x^4+h_y^4)$.  Under assumptions (A1)-(A6), there exist constants $C_3>0$ and $C_4>0$ such that
 	\begin{eqnarray*}
 		&&P\bigg(\max_{1\leq i\leq p}\sup_{{\bf u}\in \Omega_d}|\hat{\mu}_X^i({\bf u})-\mu_X^i({\bf u})|\geq \epsilon_n\bigg)\\
 		&\leq& C_3p\left(\frac{n}{\log p}\right)^{\frac{4}{4+d}}\exp\left\{-C_4n \left(\frac{\log p}{n}\right)^{\frac{d}{4+d}}\epsilon_n^2\right\},
 	\end{eqnarray*}
 	and
 	\begin{eqnarray*}
 		&&P\bigg(\max_{1\leq i\leq p}\sup_{{\bf u}\in \Omega_d}|\hat{\mu}_Y^i({\bf u})-\mu_Y^i({\bf u})|\geq \epsilon_n\bigg)\\
 		&\leq&C_3p\left(\frac{n}{\log p}\right)^{\frac{4}{4+d}}\exp\left\{-C_4n \left(\frac{\log p}{n}\right)^{\frac{d}{4+d}}\epsilon_n^2\right\}.
 	\end{eqnarray*}
 \end{lem}

 \begin{lem}\label{covariance}
 	Suppose $\epsilon_n\rightarrow 0$, $n|H_x|\epsilon_n^2\rightarrow \infty$ and there exists a large enough constant $C_h$ such that $\epsilon^2_n>C_h(h_x^4+h_y^4)$. Under assumptions (A1)-(A6), there exist constants $C_5>0$ and $C_6>0$ such that
 	\begin{eqnarray*}
 		&&P\bigg(\max_{1\leq i,j\leq p}\sup_{{\bf u}\in \Omega_d}|\hat{\sigma}_{ij}({\bf u})-\sigma_{ij}({\bf u})|\geq \epsilon_n\bigg)\\
 		&\leq& C_5p^2\left(\frac{n}{\log p}\right)^{\frac{4}{4+d}}\exp\left\{-C_6n \left(\frac{\log p}{n}\right)^{\frac{d}{4+d}}\epsilon_n^2\right\}.
 	\end{eqnarray*}
 \end{lem}
 \begin{proof}
 	We first show that there exist positive constants $c_1, c_2$ such that
 	\begin{align}\label{lemma3.1}
 		&P\bigg(\max_{1\leq i,j\leq p}\sup_{{\bf u}\in \Omega_d}\Big|\frac{1}{n_1}\sum_{k=1}^{n_1}{\bf K}_{H_x}({\bf U}_k-{\bf u})X_{ki}X_{kj}-E(X_{1i}X_{1j}|{\bf U}_1={\bf u})f({\bf u})\Big|\geq \epsilon_n\bigg)      \\
 		&\leq c_1p^2\left(\frac{n}{\log p}\right)^{\frac{4}{4+d}}\exp\left\{-c_2n \left(\frac{\log p}{n}\right)^{\frac{d}{4+d}}\epsilon_n^2\right\}. \nonumber
 	\end{align}
 	Denote $W_{kij}({\bf u})=[{\bf K}_{H_x}({\bf U}_k-{\bf u})X_{ki}X_{kj}-E{\bf K}_{H_x}({\bf U}_k-{\bf u})X_{ki}X_{kj}]/n_1$ and $M=\max_{1\leq l\leq p}EX_{1l}^4$. Notice that
 	\begin{eqnarray*}
 		EW_{kij}({\bf u})^2\leq \frac{4{ K}_1^{2d}}{n_1^2|H_x^2|}EX_{ki}^2X_{kj}^2\leq  \frac{2K_1^{2d}}{n_1^2|H_x^2|}E(X_{ki}^4+X_{kj}^4)\leq \frac{4K_1^{2d}M}{n_1^2|H_x^2|}.
 	\end{eqnarray*}
 	For any ${\bf u}\in \Omega_d$, using Markov's inequality we have, for any $0<t<\frac{n_1|H_x|}{2K_1^dM^{1/2}}$,
 	\begin{align}\label{lemma3.2}
 		&P\left(\Big|\frac{1}{n_1}\sum_{k=1}^{n_1}{\bf K}_{H_x}({\bf U}_k-{\bf u})X_{ki}X_{kj}-E{\bf K}_{H_x}({\bf U}_k-{\bf u})X_{ki}X_{kj}\Big|>\epsilon_n \right)   \\
 		\leq&  2 \exp\{-t\epsilon_n\}\Pi_{k=1}^{n_1}E  \exp\{tW_{kij}({\bf u})\}       \nonumber\\
 		\leq& 2 \exp\{-t\epsilon_n\}\Pi_{k=1}^{n_1}\{1+t^2EW_{kij}({\bf u})^2 \}  \nonumber\\
 		\leq& 2\exp\left\{-t\epsilon_n+\sum_{k=1}^{n_1}t^2E W_{kij}({\bf u})^2\right\} \nonumber\\
 		\leq& 2\exp\left\{-t\epsilon_n+\frac{Ct^2}{n_1|H_x|}\right\},\nonumber
 	\end{align}
 	for some large enough constant $C$. Here in the last step we have used the fact that $Var(W_{kij}({\bf u}))=\emph{O}(n_1^{-2}|H_x|^{-1})$.
 	By setting $t=(2C)^{-1}n_1|H_x|\epsilon_n$, we have:
 	\begin{align}\label{lemma3.3}
 		&P\left(\Big|\frac{1}{n_1}\sum_{k=1}^{n_1}{\bf K}_{H_x}({\bf U}_k-{\bf u})X_{ki}X_{kj}-E{\bf K}_{H_x}({\bf U}_k-{\bf u})X_{ki}X_{kj}\Big|>\epsilon_n \right)   \\
 		\leq& 2\exp\left\{-\frac{\epsilon_n^2n_1|H_x|}{4C}\right\}.\nonumber
 	\end{align}
 	
 	Again, without loss of generality, assume that $\Omega_d=[a_1,b_1]\times\cdots[a_d,b_d]$ and let $q_1,\ldots,q_d$,  $\omega_{i_1,\ldots,i_d}$ and ${\bf u}_{i_1,\ldots,i_d}$ for $1\leq i_j\leq q_j, 1\leq j\leq d$ be defined as in the proof of Lemma \ref{kernel}.
 	
 	Note that
 	\begin{align}\label{lemma3.4}
 		&\sup_{{\bf u}\in \Omega_d}\Big|\frac{1}{n_1}\sum_{k=1}^{n_1}{\bf K}_{H_x}({\bf U}_j-{\bf u})X_{ki}X_{kj}-E{\bf K}_{H_x}({\bf U}_j-{\bf u})X_{ki}X_{kj}\Big|  \\
 		\leq&
 		\max_{1\leq i_j \leq q_j, j=1,\ldots,d}\Big|\frac{1}{n_1}\sum_{k=1}^{n_1}{\bf K}_{H_x}({\bf U}_j-{\bf u}_{i_1,\ldots,i_d})X_{ki}X_{kj} \nonumber \\
 		&-E{\bf K}_{H_x}({\bf U}_j-{\bf u}_{i_1,\ldots,i_d})X_{ki}X_{kj}\Big| \nonumber\\
 		&+\max_{1\leq i_j \leq q_j, j=1,\ldots,d}\sup_{{\bf u}\in\omega_{i_1,\ldots,i_d}}\Big|\frac{1}{n_1}\sum_{k=1}^{n_1}{\bf K}_{H_x}({\bf U}_j-{\bf u})X_{ki}X_{kj}\nonumber \\
 		&-\frac{1}{n_1}\sum_{k=1}^{n_1}{\bf K}_{H_x}({\bf U}_j- {\bf u}_{i_1,\ldots,i_d})X_{ki}X_{kj}  \nonumber\\
 		&-[E{\bf K}_{H_x}({\bf U}_i-{\bf u})X_{ki}X_{kj}-E{\bf K}_{H_x}({\bf U}_i-{\bf u}_{i_1,\ldots,i_d})X_{ki}X_{kj}] \Big|.\nonumber
 	\end{align}
 	Using the mean value theorem we have, there exists a random scalars $0\leq R_{i_1,\ldots,i_d}\leq 2$ depending on ${\bf U}_1,\ldots, {\bf U}_{n_1}$ such that
 	\begin{align*}
 		&\sup_{{\bf u}\in\omega_{i_1,\ldots,i_d}}\Big|\frac{1}{n_1}\sum_{k=1}^{n_1}{\bf K}_{H_x}({\bf U}_j-{\bf u})X_{ki}X_{kj}-\frac{1}{n_1}\sum_{k=1}^{n_1}{\bf K}_{H_x}({\bf U}_j-{\bf u}_{i_1,\ldots,i_d})X_{ki}X_{kj} \Big|  \nonumber\\
 		&\leq \frac{2{\rm tr}(H^3_x)}{n_1} \sup_{u\in\omega_{i_1,\ldots,i_d}}\Big|\sum_{k=1}^{n_1}{\bf K}'_{H_x}({\bf U}_j-{\bf u}_{i_1,\ldots,i_d}+R_{i_1,\ldots,i_d}{\rm diag}(H_x^4))X_{ki}X_{kj}  \Big| \nonumber\\
 		&\leq\frac{ K_2{\rm tr}(H^3_x)}{n_1}\sum_{k=1}^{n_1}(X_{ki}^2+X_{kj}^2).
 	\end{align*}

 	Note that assumption (A5) implies that there exists a constant $M_1<\infty$ such that for any $1\leq i\leq p$,
 	$$\sup_{{\bf u}\in \Omega_d}|\mu_{X}^i({\bf u})|\leq M_1,~~~\sup_{{\bf u}\in \Omega_d}|\mu_{Y}^i({\bf u})|\leq M_1,~~~\sup_{{\bf u}\in \Omega_d}|\sigma_{ii}({\bf u})|\leq M_1.$$
 	By verifying the conditions of Bernstein's inequality (see for example \cite{LinBai}), we have that
 	\begin{eqnarray*}
 		P\left(\frac{1}{n_1}\sum_{k=1}^{n_1}X_{ki}^2 >M_1^2+M_1+1\right)=b_1\exp\{-b_2n_1\},
 	\end{eqnarray*}
 	for some positive constants $b_1, b_2$. Therefore, with probability greater than $1-b_1\exp\{-b_2n_1\}$,
 	\begin{align}\label{lemma3.5}
 		&\sup_{{\bf u}\in\omega_{i_1,\ldots,i_d}}\Big|\frac{1}{n_1}\sum_{k=1}^{n_1}{\bf K}_{H_x}({\bf U}_j-{\bf u})X_{ki}X_{kj}-\frac{1}{n_1}\sum_{k=1}^{n_1}{\bf K}_{H_x}({\bf U}_j-{\bf u}_{i_1,\ldots,i_d})X_{ki}X_{kj} \Big|    \\&\leq \frac{ K_2{\rm tr}(H^3_x)(M_1^2+M_1+1)}{n_1}.\nonumber
 	\end{align}
 	Clearly, $b_1\exp\{-b_2n_1\}$ is negligible comparing to the right hand side of (\ref{lemma3.1}).
 	On the other hand, by conditional on ${\bf U}_k$ first, we obtain:
 	\begin{align}\label{lemma3.6}
 		\max_{1\leq i_j \leq q_j, j=1,\ldots,d}\sup_{{\bf u}\in\omega_{i_1,\ldots,i_d}}&[E{\bf K}_{H_x}({\bf U}_k-{\bf u})X_{ki}X_{kj} \\&-E{\bf K}_{H_x}({\bf U}_i-{\bf u}_{i_1,\ldots,i_d})X_{ki}X_{kj}]   =\emph{O}({\rm tr}(H_x^4)).\nonumber
 	\end{align}
 	
 	Combining (\ref{lemma3.3}), (\ref{lemma3.4}), (\ref{lemma3.5}) and (\ref{lemma3.6}) we have:
 	\begin{align*}
 		&P\left(\sup_{{\bf u}\in \Omega_d}\Big|\frac{1}{n_1}\sum_{k=1}^{n_1}{\bf K}_{H_x}({\bf U}_k-{\bf u})X_{ki}X_{kj}-E{\bf K}_{H_x}({\bf U}_k-{\bf u})X_{ki}X_{kj}\Big|>\epsilon_n/2 \right) \\
 		&\leq c_3\left(\frac{n}{\log p}\right)^{\frac{4}{4+d}}\exp\left\{-c_4n \left(\frac{\log p}{n}\right)^{\frac{d}{4+d}}\epsilon_n^2\right\}, \nonumber
 	\end{align*}
 	for some constants $c_3>0, c_4>0$. Here in the last step we have used Assumption (A4). This together with the following well known result:
 	\begin{eqnarray*}
 		\sup_{{\bf u}\in \Omega_d}|E{\bf K}_{H_x}({\bf U}_k-{\bf u})X_{ki}X_{kj}-E(X_{1i}X_{1j}|{\bf U}_1={\bf u})f({\bf u})|=\emph{O}({\rm tr}(H^2))
 	\end{eqnarray*}
 	proves (\ref{lemma3.1}). Let $\hat{\sigma}_{ij}^X({\bf u})$ be the $(i,j)$th element of $\hat{\Sigma}_X({\bf u})$ defined as in (\ref{sigmahatx}). Using Lemma \ref{kernel} and  (\ref{lemma3.1}), it can be shown that there exist positive constants $c_5, c_6$ such that
 	\begin{eqnarray}\label{lemma3.7}
 		&&P\bigg(\max_{1\leq i,j\leq p}\sup_{{\bf u}\in \Omega_d}|\hat{\sigma}_{ij}^X({\bf u})-\sigma_{ij}^X({\bf u})|\geq \epsilon_n\bigg) \\
 		&&\leq c_1p^2\left(\frac{n}{\log p}\right)^{\frac{4}{4+d}}\exp\left\{-c_2n \left(\frac{\log p}{n}\right)^{\frac{d}{4+d}}\epsilon_n^2\right\}. \nonumber
 	\end{eqnarray}
 	Similarly let $\hat{\sigma}_{ij}^Y({\bf u})$ be the $(i,j)$th element of $\hat{\Sigma}_Y({\bf u})$ defined as in (\ref{sigmahaty}). we have that there exist positive constants $c_7, c_8$ such that
 	\begin{eqnarray}\label{lemma3.8}
 		&&P\bigg(\max_{1\leq i,j\leq p}\sup_{{\bf u}\in \Omega_d}|\hat{\sigma}_{ij}^Y({\bf u})-\sigma_{ij}^Y({\bf u})|\geq \epsilon_n\bigg) \\
 		&\leq& c_7p^2\left(\frac{n}{\log p}\right)^{\frac{4}{4+d}}\exp\left\{-c_8n \left(\frac{\log p}{n}\right)^{\frac{d}{4+d}}\epsilon_n^2\right\}. \nonumber
 	\end{eqnarray}
 	Lemma \ref{covariance} is then proved by (\ref{lemma3.7}), (\ref{lemma3.8}) and the definition of $\hat{\sigma}_{ij}({\bf u})$.
 \end{proof}

 Note that when $n<p$ and $\frac{\log p}{n}\rightarrow 0$, Lemmas \ref{mean} and \ref{covariance} are true for $\epsilon_n=M\left(\frac{\log p}{n}\right)^{\frac{2}{4+d}}$, where $M>0$ is a large enough constant. The next lemma shows that the true $\beta({\bf u})=\Sigma^{-1}[\mu_X({\bf u})-\mu_Y({\bf u})]$ belongs to the feasible set of (\ref{hatbeta}) with overwhelming probability uniformly in ${\bf u}\in\Omega_d$.

 \begin{lem}\label{feasible}
 	Under assumptions (A1)-(A6),
 	for any constant $M>0$, by choosing
 	$${\lambda}_n=C\left(\frac{\log p}{n}\right)^{\frac{2}{4+d}}\sup_{{\bf u}\in \Omega_d}\Delta({\bf u}),$$
 	for some constant $C$
 	large enough, we have with probability greater than
 	$1-{O}(p^{-M})$,
 	\begin{eqnarray*}
 		\sup_{{\bf u}\in\Omega_d}|\hat{\Sigma}({\bf u})\beta({\bf u})-[\hat{\mu}_X({\bf u})-\hat{\mu}_Y({\bf u})]|_{\infty}\leq {\lambda}_{n}.
 	\end{eqnarray*}
 \end{lem}
 \begin{proof}
 	
 	By
 	Lemma \ref{mean} we have, for any constant $M>0$, there exists a positive
 	constant $c_1>0$ large enough, such that
 	\begin{equation}\label{lemma2_1}
 		P\left(\sup_{{\bf u}\in \Omega_d}|\hat{\mu}_X({\bf u})-\hat{\mu}_Y({\bf u})-\mu_X({\bf u})+\mu_Y({\bf u})|_{\infty}\geq c_1\left(\frac{\log p}{n}\right)^{\frac{2}{4+d}}\right)\leq p^{-M}.
 	\end{equation}
 	On the other hand, using similar arguments as in the proofs of Lemma \ref{kernel}, we have there exists $c_2>0$ such that,
 	\begin{eqnarray}\label{neweq1}
 		P&\bigg(&\sup_{{\bf u}\in \Omega_d} \left|\frac{[\sum_{j=1}^{n_1}{\bf K}_{H_x}({\bf U}_j-{\bf u})X_{j}^T\beta({\bf u})]}{\sum_{j=1}^{n_1}{\bf K}_{H_x}({\bf U}_j-{\bf u})} -\mu_X({\bf u})^T\beta({\bf u})\right|  \\
 		&&\geq c_2\sup_{{\bf u}\in\Omega_d}|\beta({\bf u})|_2\left(\frac{\log p}{n}\right)^{\frac{2}{4+d}} \bigg)  \leq p^{-M}.    \nonumber
 	\end{eqnarray}

 	
 	Similar to \eqref{neweq1}, from the proofs of Lemma \ref{covariance}, it can be shown that, there exists constant $c_3>0$ such that, for $i=1,\ldots, p$,
 	\begin{eqnarray}\label{neweq2}
 		&&P\bigg(\left| \sup_{{\bf u}\in\Omega_d}\frac{\sum_{j=1}^{n_1}{\bf K}_{H_x}({\bf U}_j-{\bf u})X_{ji}X_j^T\beta({\bf u})}{\sum_{j=1}^{n_1}{\bf K}_{H_x}({\bf U}_j-{\bf u})}- EX_{ji}X_j^T\beta({\bf u})|{\bf U}_j={\bf u}
 		\right|     \\
 		&& \geq
 		c_3\sup_{{\bf u}\in\Omega_d}|\beta({\bf u})|_2\left(\frac{\log p}{n}\right)^{\frac{2}{4+d}} \bigg)   \leq
 		p^{-M-1}.  \nonumber
 	\end{eqnarray}
 	Let $\hat{\Sigma}_X({\bf u})_{i,\cdot}$ and $\Sigma({\bf u})_{i,\cdot}$ be the $i$th row of $\hat{\Sigma}_X({\bf u})$ and $\Sigma({\bf u})$ respectively. By combining \eqref{neweq1} and \eqref{neweq2} we have there exist constants $c_4>0, c_5>0$ such that
 	\begin{eqnarray}\label{lemma2_2}
 		&& P\left(\sup_{{\bf u}\in\Omega_d}|(\hat{\Sigma}_X({\bf u})-\Sigma({\bf u}))\beta({\bf u})|_{\infty}
 		\geq
 		c_2\sup_{{\bf u}\in\Omega_d}|\beta({\bf u})|_2\left(\frac{\log p}{n}\right)^{\frac{2}{4+d}} \right)    \\
 		&\leq& \sum_{i=1}^{p} P\left(\sup_{{\bf u}\in\Omega_d}|(\hat{\Sigma}({\bf u})_{i,\cdot}-\Sigma({\bf u})_{i,\cdot})\beta({\bf u})|
 		\geq
 		c_2\sup_{{\bf u}\in\Omega_d}|\beta({\bf u})|_2\left(\frac{\log p}{n}\right)^{\frac{2}{4+d}} \right)\nonumber \\
 		&\leq& c_5p^{-M}. \nonumber
 	\end{eqnarray}
 	Similarly, we have
 	\begin{eqnarray*}
 		P\left(\sup_{{\bf u}\in\Omega_d}|(\hat{\Sigma}_Y({\bf u})-\Sigma({\bf u}))\beta({\bf u})|_{\infty}
 		\geq
 		c_2\sup_{{\bf u}\in\Omega_d}|\beta({\bf u})|_2\left(\frac{\log p}{n}\right)^{\frac{2}{4+d}} \right)
 		\leq c_5p^{-M}
 	\end{eqnarray*}
 	The lemma is proved by combining the above two inequalities with (\ref{lemma2_1}), (\ref{lemma2_2}), the following inequality:
 	\begin{align*}
 		&\sup_{{\bf u}\in \Omega_d}|\hat{\Sigma}({\bf u})\beta({\bf u})-[\hat{\mu}_X({\bf u})-\hat{\mu}_Y({\bf u})]|_{\infty} \\
 		&\leq
 		\sup_{{\bf u}\in\Omega_d}|(\hat{\Sigma}({\bf u})-\Sigma({\bf u}))\beta({\bf u})|_{\infty}
 		+\sup_{{\bf u}\in \Omega_d}|\hat{\mu}_X({\bf u})-\hat{\mu}_Y({\bf u})-\mu_X({\bf u})+\mu_Y({\bf u})|_{\infty},
 	\end{align*}
 	and the fact that $\lambda|\beta({\bf u})|_2^2\geq   \Delta_{p}^2({\bf u})>\lambda^{-1}|\beta({\bf u})|_2^2$, where $\lambda$ is defined as in Assumption (A5).
 \end{proof}

 \noindent
 {\bf Proof of Theorem \ref{theorem1}}.
 We first of all derive upper bounds for
 
 (i) $ \sup_{{\bf u}\in \Omega_d}|(\hat{\mu}_X({\bf u})-\mu_X({\bf u}))^T\hat{\beta}({\bf u}) |/\Delta_p({\bf u})$ and \\
 $~~~~~~~~\sup_{{\bf u}\in \Omega_d}|(\hat{\mu}_Y({\bf u})-\mu_Y({\bf u}))^T\hat{\beta}({\bf u}) |/\Delta_p({\bf u})$,
 
 (ii) $\sup_{{\bf u}\in \Omega_d}|(\hat{\mu}_X({\bf u})-\hat{\mu}_Y({\bf u}))^T\hat{\beta}({\bf u})-
 (\mu_X({\bf u})-\mu_Y({\bf u}))^T\beta({\bf u})|/\Delta_p({\bf u})$,
 
 (iii) $\sup_{{\bf u}\in \Omega_d} |\hat{\beta}({\bf u})^T\Sigma({\bf u})\hat{\beta}({\bf u}) - (\mu_X({\bf u})-\mu_Y({\bf u}))^T\beta({\bf u}) |/\Delta^2_p({\bf u})$.

 (i) By Lemma \ref{mean}, there exists a constant $C_1>0$ large enough such that for any $M>0$, with probability larger than $1-\emph{O}(p^{-M})$,
 \begin{equation}\label{i1}
 	\sup_{{\bf u}\in \Omega_d}|\hat{\mu}_X({\bf u})-\mu_X({\bf u})|_{\infty}\leq C_1\left(\frac{\log p}{n}\right)^{\frac{2}{4+d}},
 \end{equation}
 and
 \begin{equation}\label{i2}
 	\sup_{{\bf u}\in \Omega_d}|\hat{\mu}_Y({\bf u})-\mu_Y({\bf u})|_{\infty}\leq C_1\left(\frac{\log p}{n}\right)^{\frac{2}{4+d}}.
 \end{equation}
 Together with the definition of $\hat{\beta}({\bf u})$ and Lemma \ref{feasible}  we have, with probability larger than $1-\emph{O}(p^{-M})$,
 \begin{equation}\label{i00}
 	\sup_{{\bf u}\in \Omega_d}|(\hat{\mu}_X({\bf u})-\mu_X({\bf u}))^T\hat{\beta}({\bf u}) |/\Delta_p({\bf u})\leq C_1\left(\frac{\log p}{n}\right)^{\frac{2}{4+d}} \sup_{{\bf u}\in \Omega_d}\frac{|\beta({\bf u})|_1}{\Delta_p({\bf u})},
 \end{equation}
 \begin{equation}\label{i01}
 	\sup_{{\bf u}\in \Omega_d}|(\hat{\mu}_Y({\bf u})-\mu_Y({\bf u}))^T\hat{\beta}({\bf u}) |/\Delta_p({\bf u})\leq C_1\left(\frac{\log p}{n}\right)^{\frac{2}{4+d}} \sup_{{\bf u}\in \Omega_d}\frac{|\beta({\bf u})|_1}{\Delta_p({\bf u})}.
 \end{equation}

 (ii) Notice that
 \begin{align}\label{ii1}
 	&|(\hat{\mu}_X({\bf u})-\hat{\mu}_Y({\bf u}))^T\hat{\beta}({\bf u})- (\mu_X({\bf u})-\mu_Y({\bf u}))^T\beta({\bf u})|  \\  \leq&
 	|(\hat{\mu}_X({\bf u})-\hat{\mu}_Y({\bf u}))^T\hat{\beta}({\bf u})- \beta({\bf u})^T\hat{\Sigma}({\bf u})\hat{\beta}({\bf u})| \nonumber\\ &+
 	|\beta({\bf u})^T\hat{\Sigma}({\bf u})\hat{\beta}({\bf u})-\beta({\bf u})^T(\hat{\mu}_X({\bf u})-\hat{\mu}_Y({\bf u}))|  \nonumber\\
 	&+|(\hat{\mu}_X({\bf u})-\hat{\mu}_Y({\bf u})-\mu_X({\bf u})+\mu_Y({\bf u}))^T\beta({\bf u})|.\nonumber
 \end{align}
 By the definition of $\hat{\beta}({\bf u})$ and Lemma \ref{feasible} we have with probability larger than $1-\emph{O}(p^{-M})$,
 \begin{align}\label{ii2}
 	& \sup_{{\bf u}\in \Omega_d}|(\hat{\mu}_X({\bf u})-\hat{\mu}_Y({\bf u}))^T\hat{\beta}({\bf u})- \beta({\bf u})^T\hat{\Sigma}({\bf u})\hat{\beta}({\bf u})|/\Delta_p({\bf u})
 	\\ \leq & \sup_{{\bf u}\in \Omega_d}|(\hat{\mu}_X({\bf u})-\hat{\mu}_Y({\bf u}))-\hat{\Sigma}({\bf u})\beta({\bf u})|_{\infty} |\hat{\beta}({\bf u})|_1/\Delta_p({\bf u})\nonumber\\
 	\leq&  \lambda_n\sup_{{\bf u}\in \Omega_d}\frac{|\beta({\bf u})|_1}{\Delta_p({\bf u})}.\nonumber
 \end{align}
 Similarly, by the definition of $\hat{\beta}({\bf u})$ we have
 \begin{align}\label{ii3}
 	&\sup_{{\bf u}\in \Omega_d} |\beta({\bf u})^T\hat{\Sigma}({\bf u})\hat{\beta}({\bf u})-\beta({\bf u})^T(\hat{\mu}_X({\bf u})-\hat{\mu}_Y({\bf u}))|/\Delta_p({\bf u}) \\
 	\leq& {\lambda_n} \sup_{{\bf u}\in \Omega_d}\frac{|\beta({\bf u})|_1}{\Delta_p({\bf u})}. \nonumber
 \end{align}

 From (\ref{ii1}), (\ref{ii2}), (\ref{ii3}) and the proofs of (\ref{i1}), (\ref{i2}), we have with probability larger than $1-\emph{O}(p^{-M})$,
 \begin{align}\label{ii}
 	&|(\hat{\mu}_X({\bf u})-\hat{\mu}_Y({\bf u}))^T\hat{\beta}({\bf u})- (\mu_X({\bf u})-\mu_Y({\bf u}))^T\beta({\bf u})|/\Delta_p({\bf u})    \\
 	\leq&
 	2\lambda_n\sup_{{\bf u}\in \Omega_d}\frac{|\beta({\bf u})|_1}{\Delta_p({\bf u})}+2 C_1\left(\frac{\log p}{n}\right)^{\frac{2}{4+d}}  \sup_{{\bf u}\in \Omega_d}\frac{|\beta({\bf u})|_1}{\Delta_p({\bf u})}.\nonumber
 \end{align}
 
 (iii) Notice that
 \begin{align}\label{iii1}
 	&|\hat{\beta}({\bf u})^T\Sigma({\bf u})\hat{\beta}({\bf u}) - (\mu_X({\bf u})-\mu_Y({\bf u}))^T\beta({\bf u}) |   \\
 	\leq& |\hat{\beta}^T(\hat{\mu}_X({\bf u})-\hat{\mu}_Y({\bf u})) -\beta({\bf u})^T(\mu_X({\bf u})-\mu_Y({\bf u})) |     \nonumber \\
 	&+ |\hat{\beta}({\bf u})^T\Sigma({\bf u})\hat{\beta}({\bf u}) -\hat{\beta}^T(\hat{\mu}_X({\bf u})-\hat{\mu}_Y({\bf u})) | .\nonumber
 \end{align}
 From the definition of $\hat{\beta}({\bf u})$ and the bounds for (ii),
 %
 we have, there exists a constant $C_2$ large enough such that with probability larger than $1-\emph{O}(p^{-M})$,
 \begin{align}\label{iii}
 	&\sup_{{\bf u}\in \Omega_d} |\hat{\beta}({\bf u})^T\Sigma({\bf u})\hat{\beta}({\bf u}) - (\mu_X({\bf u})-\mu_Y({\bf u}))^T\beta({\bf u}) |/\Delta^2_p({\bf u}) \\
 	\leq&
 	2\lambda_n\sup_{{\bf u}\in \Omega_d}\frac{|\beta({\bf u})|_1}{\Delta_p^2({\bf u})}+2 C_1\left(\frac{\log p}{n}\right)^{\frac{2}{4+d}}  \sup_{{\bf u}\in \Omega_d}\frac{|\beta({\bf u})|_1}{\Delta_p^2({\bf u})}+ \lambda_n\sup_{{\bf u}\in \Omega_d}\frac{|\beta({\bf u})|_1}{\Delta_p^2({\bf u})}.
 	\nonumber 
 \end{align}
 Combining (\ref{i00}), (\ref{i01}), \eqref{ii}, \eqref{iii} and Assumption (A5) we have, there exists large enough constants $C_3, C_4>0$, such that with probability larger than $1-\emph{O}(p^{-M})$, uniformly for any ${\bf u}\in \Omega_d$,
 \begin{align}\label{last1}
 	& \frac{(\hat{\mu}_X({\bf u})-\hat{\mu}_Y({\bf u}))^T\hat{\beta}({\bf u})}{2\sqrt{\hat{\beta}({\bf u})^T\Sigma({\bf u})\hat{\beta}({\bf u})}}+\frac{(\hat{\mu}_Y({\bf u})-\mu_Y({\bf u}))^T\hat{\beta}({\bf u})}{\sqrt{\hat{\beta}({\bf u})^T\Sigma({\bf u})\hat{\beta}({\bf u})}}  \\
 	=& \frac{(\hat{\mu}_X({\bf u})-\hat{\mu}_Y({\bf u}))^T\hat{\beta}({\bf u})/\Delta_p({\bf u})+
 		2(\hat{\mu}_Y({\bf u})-\mu_Y({\bf u}))^T\hat{\beta}({\bf u})/\Delta_p({\bf u})}{2\sqrt{\hat{\beta}({\bf u})^T\Sigma({\bf u})\hat{\beta}({\bf u})/\Delta^2_p({\bf u})}}   \nonumber\\
 	\leq&  \frac{\Delta_p({\bf u})+2\lambda_n\sup_{{\bf u}\in \Omega_d}\frac{|\beta({\bf u})|_1}{\Delta_p({\bf u})}+C_4\left(\frac{\log p}{n}\right)^{\frac{2}{4+d}} \sup_{{\bf u}\in \Omega_d}\frac{|\beta({\bf u})|_1}{\Delta_p({\bf u})}
 	}{2\sqrt{1-3\lambda_n\sup_{{\bf u}\in \Omega_d}\frac{|\beta({\bf u})|_1}{\Delta^2_p({\bf u})}-C_3\left(\frac{\log p}{n}\right)^{\frac{2}{4+d}} \sup_{{\bf u}\in \Omega_d}\frac{|\beta({\bf u})|_1}{\Delta^2_p({\bf u})}}}  \nonumber \\
 	=&\frac{\Delta_p({\bf u})}{2}\left[1+O\left( \left(\frac{\log p}{n}\right)^{\frac{2}{4+d}} \sup_{{\bf u}\in \Omega_d}\frac{|\beta({\bf u})|_1}{\Delta_p({\bf u})}\right)\right].  \nonumber
 \end{align}
 Theorem \ref{theorem1} can be proved by (\ref{assump0}) and Lemma \ref{risk}.
 \\

 \begin{lem}\label{lemma0}
 	Let $\Phi$ and $\phi$ be the cumulative distribution function and density function of a standard Gaussian random variable. For any $x\geq 1$ we have
 	\begin{eqnarray*}
 		\frac{\phi(x)}{2x}\leq \Phi(-x)\leq \frac{\phi(x)}{x}.
 	\end{eqnarray*}
 \end{lem}
 \begin{proof}
 	Using integration by parts we have for $x\geq 1$:
 	\begin{eqnarray*}
 		\Phi(-x)=-\frac{\phi(x)}{x}-\int_{x}^{+\infty}\frac{1}{u^2}\phi(u)du\leq -\frac{\phi(x)}{x}- \Phi(-x).
 	\end{eqnarray*}
 	Lemma \ref{lemma0} is then proved immediately from the above inequality.
 \end{proof}
 {\bf Remark:} Lemma \ref{lemma0} implies that $\Phi(-x)=O\bigg(\frac{\phi(x)}{x}\bigg)$ for any $x>B/2$.

 \noindent
 {\bf Proof of Theorem \ref{theorem2}}.
 
 By Lemma \ref{lemma0}, similar to \eqref{last1},  we have, uniformly in ${\bf u}\in \Omega_d$,
 \begin{align}
 	&\Phi\left(-\frac{(\hat{\mu}_X({\bf u})-\hat{\mu}_Y({\bf u}))^T\hat{\beta}({\bf u})}{2\sqrt{\hat{\beta}({\bf u})^T\Sigma({\bf u})\hat{\beta}({\bf u})}}-\frac{(\hat{\mu}_Y({\bf u})-\mu_Y({\bf u}))^T\hat{\beta}({\bf u})}{\sqrt{\hat{\beta}({\bf u})^T\Sigma({\bf u})\hat{\beta}({\bf u})}}\right) \\
 	=&\Phi\left(-\frac{\Delta_p({\bf u})}{2}\right)+
 	O\left(\left(\frac{\log p}{n}\right)^{\frac{2}{4+d}}\sup_{{\bf u}\in \Omega_d} \frac{|\beta({\bf u})|_1}{\Delta_p({\bf u})}\right) \times \nonumber\\
 	&\phi\left(-\frac{\Delta_p({\bf u})}{2}+O\left(\left(\frac{\log p}{n}\right)^{\frac{2}{4+d}}\sup_{{\bf u}\in \Omega_d}\frac{|\beta({\bf u})|_1}{\Delta_p({\bf u})}\right) \right) \nonumber\\
 	=& \Phi\left(-\frac{\Delta_p({\bf u})}{2}\right)\Bigg[1+O\bigg( \left(\frac{\log p}{n}\right)^{\frac{2}{4+d}}\sup_{{\bf u}\in \Omega_d}\Delta_p({\bf u})\sup_{{\bf u}\in \Omega_d} \frac{|\beta({\bf u})|_1}{\Delta_p({\bf u})}\times  \nonumber \\
 	& \exp\bigg\{\left(\frac{\log p}{n}\right)^{\frac{2}{4+d}}\sup_{{\bf u}\in \Omega_d}\Delta_p({\bf u}) \sup_{{\bf u}\in \Omega_d}\frac{|\beta({\bf u})|_1}{\Delta_p({\bf u})}
 	\bigg\}\bigg) \Bigg] \nonumber \\
 	=& \Phi\left(-\frac{\Delta_p({\bf u})}{2}\right)\Bigg[1+O\bigg(\left(\frac{\log p}{n}\right)^{\frac{2}{4+d}}\sup_{{\bf u}\in \Omega_d}\Delta_p({\bf u})\sup_{{\bf u}\in \Omega_d}  \frac{|\beta({\bf u})|_1}{\Delta_p({\bf u})}\bigg)\Bigg]. \nonumber
 \end{align}
 Theorem \ref{theorem2} can then be proved by (\ref{assump2}) and Lemma \ref{risk}.

 \noindent
 {\bf Proof of Theorem \ref{thminimax}}.
 For simplicity we consider the case where the dynamic factor ${\bf u}$ is of dimension $d=1$ and use the notation $u$ instead in this proof. The following proofs can be generalized for any given integer $d>0$ simply by some regular arguments. We prove \eqref{minimax1} first where the distance is defined as $d(\theta_1,\theta_2):=\sup_{{\bf u}\in \Omega_d}|T(\theta_1)-T(\theta_2)|$ with $T(\theta)=\Delta_p(u)$.
 
 \noindent
 {\sc Step 1}. Construction of the hypotheses.
 
 \noindent
 We assume that ${ u}$ is generated from $U[0,1]$.
 Consider $\Sigma({ u})=I_p$ where $I_p$ is the $p\times p$ identity matrix. We then have $\beta({u})=\mu_X({u})-\mu_Y({ u})$. We set the null hypothesis as $(\mu_X(u),\mu_Y(u),\Sigma(u))=\theta_0=(0_p,0_p, I_p)$, where $0_p$ is the $p$-dimensional vector
 of zeros. 
 Clearly we would have $\theta_0\in {\cal G}(\kappa)$. For the alternatives, let $\lfloor \cdot\rfloor$ be the largest integer function and define:
 \begin{eqnarray*}\label{alt0}
 	m_h=\lfloor h^{-1}\rfloor,~~    h= \left( \frac{ \log p }{n }\right)^{1/5}~~  u_k=\frac{k-0.5}{m_h}, \nonumber\\
 	\epsilon_k(u)=h^2K\left(\frac{u-u_k}{h}\right),~~ k=1,\ldots, m_h, ~~u\in[0,1],
 \end{eqnarray*}
 where $K: R\rightarrow [0,+\infty)$ is a kernel function such that $K\in H(2,1/2)\cap C^{\infty}(R)$ and $K(u)>0 \iff u\in(-1/2,1/2)$, and $\alpha$ is a constant such that $0<\alpha<(2e)^{-1}$.
 We set $\mu_Y=0_p$ and so $\beta(u)=\mu_X(u)$. Without loss of generality, assume that $\kappa\in Z$, the set of all integers. The parameter space ${\cal D}_1$ is then set to be
 \begin{eqnarray*}
 	{\cal D}_1&=&\{(\mu_X(u),\mu_Y(u),\Sigma(u)): \mu_Y=0_p,\Sigma=I_p, \mu_X=\epsilon_i{\bf a}
 	=\epsilon_i(a_1,\ldots,a_p)^T,\\
 	&& |{\bf a} |_0=  \kappa, a_j\in \{0,1\}, i=1,\ldots,m_h \}.
 \end{eqnarray*}
 The cardinality of ${\cal D}_1$ is then $m=m_h\binom p\kappa$. Clearly we have, for any $\theta_i \in {\cal D}_1, i=1\ldots, m$,
 \begin{itemize}
 	\item[(i)] $\theta_i\in {\cal G}(\kappa)$.
 	\item[(ii)] $d(\theta_0,\theta_i)=\sup_{u\in [0,1]}|\theta_i|_2=h^2\sqrt{\kappa}=\sqrt{\kappa} \left(  \frac{\alpha \log p }{n }\right)^{2/5}$.
 \end{itemize}

 \noindent
 {\sc Step 2}. Bounding the total variance
 
 \noindent
 Given $u$, we denote the density function of the multivariate standard Gaussian distribution $N(0_p, I_p)$ as $f_0$ and for any $\theta_i\in {\cal D}_1$. Recall that $\theta=(\mu_X(\bf u),\mu_Y(\bf u),\Sigma(\bf u))$. For a given $\theta=\theta_i$, we shall denote the corresponding $\mu_Y$ as $\nu_i:=\mu_Y|_{\theta=\theta_i}$ and let $f_i$ be the density of the Gaussian distribution $N(\nu_i, I_p)$. We set the weight to be $\omega_1=\cdots=\omega_m=m^{-1}$ and for any probability measures $Q,R$, we use $\chi^2(Q,R)$ to denote the $\chi^2$ divergence of $Q$ and $R$. By (2.27) in \cite{Tsybakov}, 
 we have,
 \nocite*

 \begin{eqnarray}\label{mm1}
 	&&\|\bar{P}-P_0\|_1^2  \\
 	&\leq& \chi^2(\bar{P},P_0) \nonumber \\
 	&=& \int_{[0,1]^{n}} \int_{R^{p\times n}}\frac{(\sum_{j=1}^mm^{-1}\Pi_{i=1}^{n_2}f_j(x_i) \Pi_{l=n_2+1}^{n}f_0(x_l))^2}{\Pi_{i=1}^{n}f_0(x_i)}   \nonumber\\
 	&&dx_1\cdots  dx_{n}du_1\ldots du_{n} -1. \nonumber\\
 	&=& \int_{[0,1]^{n_2}} \int_{R^{p\times n_2}}\frac{(\sum_{j=1}^mm^{-1}\Pi_{i=1}^{n_2}f_j(x_i)  )^2}{\Pi_{i=1}^{n_2}f_0(x_i)}   \nonumber\\
 	&&dx_1\cdots  dx_{n_2}du_1\ldots du_{n_2} -1. \nonumber
 \end{eqnarray}
 Note that for any $1\leq j\neq k\leq m_h$ and  $t=1,\ldots, \kappa$, 
 \begin{eqnarray*}
 	\int_{[0,1]^{n_2}} \int_{R^{t\times n_2}} \left[\frac{\Pi_{i=1}^{n_2}e^{-\frac{(x_i-\epsilon_j(u_i))^2}{2}}e^{-\frac{(x_i-\epsilon_k(u_i))^2}{2}}}{{(2\pi)^{\frac{n_2}{2}}} \Pi_{i=1}^{n_2}e^{-\frac{x_i^2}{2}}}\right]^tdx_1\cdots  dx_{n_2}du_1\ldots du_{n_2}
 	=1.
 \end{eqnarray*}
 Using some combination arguments, we thus have
 \begin{eqnarray}\label{mm2}
 	&&\int_{[0,1]^{n_2}} \int_{R^{p\times n_2}}\frac{(\sum_{j=1}^mm^{-1}\Pi_{i=1}^{n_2}f_j(x_i)  )^2}{\Pi_{i=1}^{n_2}f_0(x_i)}
 	dx_1\cdots  dx_{n_2}du_1\ldots du_{n_2} \\
 	&=&\left(1-\frac{m_h\binom p\kappa ^2}{m^2}\right)\cdot 1 \nonumber\\
 	&&+\frac{1}{m^2}\int_{[0,1]^{n_2}}\sum_{i=1}^{m_h}\sum_{j=0}^{\kappa}
 	\binom p\kappa\binom \kappa j\binom {p-\kappa} {\kappa-j}   e^{j\sum_{l=1}^{n_2}\epsilon_i^2(u_l)}  du_1\cdots du_{n_2}  \nonumber\\
 	&=&\left(1-\frac{1}{m_h}\right) +\frac{1}{m^2_h}\int_{[0,1]^{n_2}} E \sum_{i=1}^{m_h}e^{J\sum_{l=1}^{n_2}\epsilon_i^2(u_l)} du_1\cdots du_{n_2},  \nonumber
 \end{eqnarray}
 where $J$ is a random variable with the Hypergeometric distribution with parameters $(p,\kappa,\kappa)$.  On the other hand, by Lemma 3 in
 \cite{Cai2017} and the fact that $e^x\leq 1+ex$ for any $0\leq x\leq 1$, we have, when $\kappa h^4\leq 1$,
 \begin{eqnarray}\label{mm3}
 	&&\frac{1}{m^2_h}\int_{[0,1]^{n_2}} E \sum_{i=1}^{m_h}e^{J\sum_{l=1}^{n_2}\epsilon_i^2(u_l)} du_1\cdots du_{n_2} \\
 	&=&\frac{1}{m^2_h}E\sum_{i=1}^{m_h} \left(  1-\frac{1}{m_h}+\frac{1}{m_h}e^{Jh^4}\right)^{n_2} \nonumber\\
 	&\leq&  \frac{1}{m_h}E\left(1+\frac{e}{m_h}Jh^4\right)^{n_2}   \nonumber\\
 	&\leq& \frac{1}{m_h}Ee^{\frac{n_2eJh^4}{m_h}}   \nonumber\\
 	&\leq& \frac{1}{m_h} e^{\frac{\kappa^2}{p-\kappa}} \left(1-\frac{\kappa}{p}+\frac{\kappa}{p}e^{\frac{n_2eh^4}{m_h}}\right)^\kappa, \nonumber \\
 	&=& O(m_h^{-1}e^{ {\kappa^2} p^{2e\alpha-1}}). \nonumber
 \end{eqnarray}
 Here in the last step we have used the fact that $\frac{n_2eh^4}{m_h}\leq \frac{n_2h^5}{1-h}\leq 2e\alpha\log p$, $\kappa=O(p^{\gamma})$ with $\gamma<\frac{1}{2}$ and $2e\alpha<1$.
 By setting $\alpha=\frac{1-2\gamma}{2e}$, from \eqref{mm1}, \eqref{mm2} and \eqref{mm3} we immediately have,
 \begin{eqnarray*}
 	\|\bar{P}-P_0\|_1^2=O(m_h^{-1}).
 \end{eqnarray*}
 Consequently, by Lemma \ref{lecam}, we conclude that \eqref{minimax1} holds.

 \noindent
 {\sc Proof of \eqref{minimax2}.}
 
 \noindent
 Set the distance to be $d(\theta_1,\theta_2):=\sup_{{\bf u}\in \Omega_d}|T(\theta_1)-T(\theta_2)|$ with $T(\theta)=R({\bf u})$. By the assumption that $\kappa=o\left(\left(\frac{n}{\log p}\right)^{\frac{4}{4+d}}\right)$ we have that (ii) in the proof of step 1 becomes $d(\theta_0,\theta_i)=\frac{1}{2}-\Phi\left(-\frac{h^2\sqrt{\kappa}}{2}\right)\leq \sqrt{\kappa} \left(  \frac{\alpha \log p }{n }\right)^{2/5}$. The rest of the proofs are the same as those for \eqref{minimax1}.

 \vspace{2mm}
 \noindent
 {\bf Proof of Theorem \ref{misethm}.}
 
 We first of all show that 
 \begin{eqnarray}\label{bound1}
 	E\|\hat{\mu}_{X,-i}({\bf U}_i)-{\mu}_X({\bf U}_i)\|_\infty^2= O\left(\Big(\frac{\log p}{n}\Big)^{\frac{2}{2+d}}\right),
 \end{eqnarray}
 and 
 \begin{eqnarray}\label{bound2}
 	E\|\hat{\Sigma}_{X,-i}({\bf U}_i)-{\Sigma}({\bf U}_i)\|_\infty^2= O\left(\Big(\frac{\log p}{n}\Big)^{\frac{2}{2+d}}\right).
 \end{eqnarray}
 By Lemma \ref{mean},  we have
 \begin{eqnarray*}
 	&&	E\|\hat{\mu}_{X,-i}({\bf U}_i)-{\mu}_X({\bf U}_i)\|_\infty^2\\
 	&\leq& \epsilon_n^2
 	+\int_{\epsilon_n}^\infty C_3p\left(\frac{n}{\log p}\right)^{\frac{4}{4+d}}\exp\left\{-C_4n \left(\frac{\log p}{n}\right)^{\frac{d}{4+d}}x^2\right\}dx\\
 	&\leq&  \epsilon_n^2+ C_3p\left(\frac{n}{\log p}\right)^{\frac{4}{4+d}}
 	\left[ -2C_4n \left(\frac{\log p}{n}\right)^{\frac{d}{4+d}}xx\right]^{-1}
 	\exp\left\{-C_4n \left(\frac{\log p}{n}\right)^{\frac{d}{4+d}}x^2\right\}\Bigg|_{\epsilon_n}^\infty.
 \end{eqnarray*}
 \eqref{bound1} is then proved by choosing $\epsilon_n=C\left(\frac{\log p}{n}\right)^{\frac{2}{2+d}}$ for a large enough constant $C$.
 \eqref{bound2} can be similarly proved using  Lemma \ref{covariance}. Now we proceed to prove the theorem.  
 Note that
 \begin{eqnarray}\label{mise0}
 	&&	E  \big\| \big(X_{i}-\hat{\mu}_{X,-i}({\bf U}_i)\big) \big(X_{i}-\hat{\mu}_{X,-i}({\bf U}_i)\big)^T-\hat{\Sigma}_{X,-i}({\bf U}_i)\big\|_F^2 \\
 	&=&E  \big\| \Sigma({\bf U}_i)-\hat{\Sigma}_{X,-i}({\bf U}_i)\big\|_F^2 +E\big\| \big(X_{i}-\hat{\mu}_{X,-i}({\bf U}_i)\big) \big(X_{i}-\hat{\mu}_{X,-i}({\bf U}_i)\big)^T- \Sigma({\bf U}_i)\big\|_F^2 \nonumber\\
 	&&+2E tr[ \Sigma({\bf U}_i)-\hat{\Sigma}_{X,-i}({\bf U}_i)] 
 	[ \big(X_{i}-\hat{\mu}_{X,-i}({\bf U}_i)\big) \big(X_{i}-\hat{\mu}_{X,-i}({\bf U}_i)\big)^T- \Sigma({\bf U}_i)]. \nonumber
 \end{eqnarray}
 On the other hand, we have
 \begin{eqnarray}\label{mise1}
 	p^{-2}	E  \big\| \Sigma({\bf U}_i)-\hat{\Sigma}_{X,-i}({\bf U}_i)\big\|_F^2 =  r(H_x),
 \end{eqnarray}
 \begin{eqnarray}\label{mise2}
 	&&p^{-2}	E\big\| \big(X_{i}-\hat{\mu}_{X,-i}({\bf U}_i)\big) \big(X_{i}-\hat{\mu}_{X,-i}({\bf U}_i)\big)^T- \Sigma({\bf U}_i)\big\|_F^2 \\
 	&=&p^{-2}	E\big\| \big(X_{i}-{\mu}_X({\bf U}_i)\big) \big(X_{i}-{\mu}_X({\bf U}_i)\big)^T- \Sigma({\bf U}_i)\big\|_F^2\nonumber \\
 	&&+p^{-2}	E\big\|  \big(X_{i}-\hat{\mu}_{X,-i}({\bf U}_i)\big) \big(X_{i}-\hat{\mu}_{X,-i}({\bf U}_i)\big)^T-\big(X_{i}-{\mu}_X({\bf U}_i)\big) \big(X_{i}-{\mu}_X({\bf U}_i)\big)^T\big\|_F^2 \nonumber\\
 	&&+2p^{-2}	E tr[ \big(X_{i}-{\mu}_X({\bf U}_i)\big) \big(X_{i}-{\mu}_X({\bf U}_i)\big)^T- \Sigma({\bf U}_i)] \nonumber\\
 	&&\cdot[\big(X_{i}-\hat{\mu}_{X,-i}({\bf U}_i)\big) \big(X_{i}-\hat{\mu}_{X,-i}({\bf U}_i)\big)^T-\big(X_{i}-{\mu}_X({\bf U}_i)\big) \big(X_{i}-{\mu}_X({\bf U}_i)\big)^T]\nonumber\\
 	&=&p^{-2}	E\big\| \big(X_{i}-{\mu}_X({\bf U}_i)\big) \big(X_{i}-{\mu}_X({\bf U}_i)\big)^T- \Sigma({\bf U}_i)\big\|_F^2\nonumber\\
 	&&
 	+O(E\|\hat{\mu}_{X,-i}({\bf U}_i)-{\mu}_X({\bf U}_i)\|_\infty^2) \nonumber\\
 	&=&p^{-2}	E\big\| \big(X_{i}-{\mu}_X({\bf U}_i)\big) \big(X_{i}-{\mu}_X({\bf U}_i)\big)^T- \Sigma({\bf U}_i)\big\|_F^2 +O\left(\Big(\frac{\log p}{n}\Big)^{\frac{2}{2+d}}\right),\nonumber
 \end{eqnarray}
 and, by the fact that $\Sigma({\bf U}_i)-\hat{\Sigma}_{X,-i}({\bf U}_i)$ and $
 \big(X_{i}-{\mu}_{X}({\bf U}_i)\big) \big(X_{i}-{\mu}_{X}({\bf U}_i)\big)^T- \Sigma({\bf U}_i) $ are conditionally independent given ${\bf U}_i$, we have
 \begin{eqnarray}\label{mise3}
 	&&p^{-2}	E tr[ \Sigma({\bf U}_i)-\hat{\Sigma}_{X,-i}({\bf U}_i)] 
 	[ \big(X_{i}-\hat{\mu}_{X,-i}({\bf U}_i)\big) \big(X_{i}-\hat{\mu}_{X,-i}({\bf U}_i)\big)^T- \Sigma({\bf U}_i)] ~~~~~~~~~~~\\
 	&=&p^{-2}	E tr[ \Sigma({\bf U}_i)-\hat{\Sigma}_{X,-i}({\bf U}_i)] 
 	[ \big(X_{i}-\hat{\mu}_{X,-i}({\bf U}_i)\big) \big(X_{i}-\hat{\mu}_{X,-i}({\bf U}_i)\big)^T\nonumber\\
 	&&-    \big(X_{i}-{\mu}_{X}({\bf U}_i)\big) \big(X_{i}-{\mu}_{X}({\bf U}_i)\big)^T ] \nonumber \\
 	&=&O(	E\|\hat{\Sigma}_{X,-i}({\bf U}_i)-{\Sigma}({\bf U}_i)\|_\infty^2)+O(E\|\hat{\mu}_{X,-i}({\bf U}_i)-{\mu}_X({\bf U}_i)\|_\infty^2)\nonumber\\
 	&=&O\left(\Big(\frac{\log p}{n}\Big)^{\frac{2}{2+d}}\right).\nonumber
 \end{eqnarray}
 The theorem is then proved by combining \eqref{mise0}, \eqref{mise1}, \eqref{mise2} and \eqref{mise3}.
 
 \section*{Acknowledgements}
 We are grateful to Prof. Holger Dette, an associate editor and two referees for their constructive comments which have led to an improved paper.

\end{document}